\documentclass[sn-mathphys-num,pdflatex]{sn-jnl}

\usepackage{fix-cm}
\usepackage{graphicx}
\usepackage{amsmath,amssymb,amsfonts}
\usepackage{amsthm}
\usepackage[title]{appendix}
\usepackage{booktabs}
\usepackage{mathtools}
\usepackage{enumitem}
\usepackage{longtable}
\usepackage{array}

\DeclarePairedDelimiter\abs{\lvert}{\rvert}

\renewcommand{\le}{\leqslant}

\renewcommand{\ge}{\geqslant}

\newcommand{\cC}{\mathcal{C}}
\newcommand{\cD}{\mathcal{D}}
\newcommand{\cT}{\mathcal{T}}
\newcommand{\B}{\mathcal{B}}
\newcommand{\LPC}{\operatorname{LPC}_{\infty}}
\newcommand{\lmax}{\ell_{\max}}
\newcommand{\Live}{\mathcal{L}}

\newcommand{\ko}{K^{\star}}
\newcommand{\dko}{K_{\mathrm{det}}^{\star}}
\newcommand{\lc}{\operatorname{lcm}}
\newcommand{\dH}{d_{\mathrm H}}
\newcommand{\Dmin}{D_{\min}}
\newcommand{\Tball}{\mathsf{T}}
\newcommand{\wt}{\operatorname{wt}}
\newcommand{\comp}{\operatorname{comp}}
\newcommand{\Pref}{\operatorname{Pref}}
\newcommand{\prefx}{\operatorname{pref}}

\theoremstyle{thmstyleone}
\newtheorem{theorem}{Theorem}
\newtheorem{proposition}[theorem]{Proposition}
\newtheorem{corollary}[theorem]{Corollary}
\newtheorem{lemma}[theorem]{Lemma}

\newtheorem{conjecture}[theorem]{Conjecture}

\theoremstyle{thmstyletwo}
\newtheorem{example}{Example}

\theoremstyle{thmstylethree}
\newtheorem{definition}{Definition}

\hypersetup{bookmarksdepth=subsubsection}
\begin{document}

\title[Extinction Depth and \(q\)-ary Error-Correcting Codes]%
{Extinction Depth and $q$-ary Error-Correcting Codes for the Limited Permutation Channel}

\author[1]{\fnm{Noam} \sur{Ben Shimon}}
\email{noamben@campus.technion.ac.il}

\author[2]{\fnm{Aryeh~Lev} \sur{Zabokritskiy (Yohananov)}}
\email{yuhanalev@telhai.ac.il}

\affil[1]{%
\orgdiv{Department of Computer Science},
\orgname{Technion -- Israel Institute of Technology},
\orgaddress{\city{Haifa}, \country{Israel}}
}

\affil[2]{%
\orgdiv{Department of Computer Science},
\orgname{MIGAL -- Galilee Research Institute/Tel-Hai University of Kiryat Shmona and the Galilee},
\orgaddress{\city{Kiryat Shmona}, \country{Israel}}
}

\abstract{
In the radius-one limited permutation channel, errors consist of disjoint
adjacent transpositions.  A correcting code must separate distinct codewords:
their error balls may not contain a common received word.  Hamming distance
does not ensure this, because disjoint swaps can make words differing in many
positions confusable.  For block-concatenation codes, earlier work tested each
possible collision only through the longer initial block.  This sufficient
condition is not necessary: we exhibit a valid ternary block set that it does
not certify.  We introduce extinction depth, which tracks unresolved
first-block pairs through later block extensions, and prove that their
extinction at one common finite horizon certifies correction at every length.
The criterion gives explicit \(q=3\) and \(q=5\) block sets with rates above
\(0.6777475\) and \(0.6694926\).  No block-set-independent depth bound exists:
we give an exact linear family, verify a quadratic formula for every
\(3\le k\le60\), and derive a polynomial-time finite-graph test.  For growing
alphabets, the normalized correction loss lies between \(\ln\varphi\) and
\(\ln(1.82560995)\), while explicit finite-length covers improve
finite-alphabet upper bounds.  We develop a parallel directed-extinction
theory for detection, including an all-length block criterion, a set unresolved
by the earlier test, and exact corridor depth \(4k\).  Stable type lifting
yields optimal first-order loss \(\sqrt2\), and weak-zigzag codes improve the
\(q=3,4\) lower bounds.  Finally, window restriction, cancellation, and a
bounded pending-input frontier extend the criterion and its finite-state
verification to every fixed displacement radius \(r\).
}



\maketitle
\section{Introduction}
\label{sec:intro}

Permutation and reordering channels model communication systems in which a
transmitted word is received with the same symbol multiplicities but in a
perturbed order.  When arbitrary reorderings are allowed, the receiver retains
only the multiset of transmitted symbols, leading to codes in multiset
spaces~\cite{KovacevicTan18}.  A separate line of work takes the codewords
themselves to be permutations.  For background on metrics on permutations, see
Deza and Huang~\cite{DezaHuang98}; permutation arrays and rank-modulation codes
under the Chebyshev, or \(\ell_\infty\), metric were developed in
\cite{KloveEtAl10,TamoSchwartz10} and continue to be studied in
\cite{BHMS24}.  The channel considered here differs from both settings: its
inputs are arbitrary \(q\)-ary words, so symbols may repeat, but their
reordering is locally bounded.

Langberg, Schwartz, and Yaakobi introduced the \(\ell_\infty\)-limited
permutation channel \(\LPC(r)\), where no symbol may move by more than \(r\)
positions~\cite{LSY17}.  We focus on \(r=1\), for which one channel use applies
a set of pairwise disjoint adjacent transpositions.  Adjacent-transposition
errors have also been studied as one or a bounded number of successive edits,
sometimes together with deletions or under an asymmetric swap
rule~\cite{GabrysYaakobiMilenkovic18,KhuatKim23,WangVuTan25}.  Those models may
allow swaps to overlap over time or impose a different error budget; they are
not the same as one simultaneous layer of disjoint swaps.  Kova\v{c}evi\'c,
Goyal, and Kiah recently studied the latter model directly, deriving
asymptotic bounds for a prescribed number of swaps and a block construction
that corrects every admissible error~\cite{KGK25}.

For this channel, correcting every admissible error means keeping each pair of
distinct codewords distinguishable after transmission.  A \emph{collision}
occurs when two codewords can produce the same received word, equivalently
when their sets of possible channel outputs intersect.  Hamming distance alone
does not characterize this requirement: simultaneous disjoint transpositions
can make even words differing in many symbol positions confusable.

Block and string-concatenation constructions already underlie the recursive
families of Langberg, Schwartz, and Yaakobi and the later constructions of
Chee et al.~\cite{LSY17,CKLNVZ16}.  In~\cite{PriorWork2026}, this approach was
formulated through a finite set \(P\) of \(q\)-ary blocks generating
\(C_n(P)\), the length-\(n\) words formed by free concatenation of blocks from
\(P\).  This framework describes codes of all lengths by one finite object,
but correctness is not a blockwise property: a transposition across a block
boundary may obscure the parsing and carry a possible collision through
several later blocks.  To certify such a block set, that work used a local
sufficient test.  For each pair of distinct possible first blocks, the test
stops when the longer initial block has ended; only the shorter side is
extended as needed to reach that point.  It does not continue both candidate
words through later blocks.  Passing the test certifies that every \(C_n(P)\)
is correcting, but failure means only that the test is inconclusive, not that
a collision exists.  We restate this test as the two-stage criterion in
Theorem~\ref{thm:two-stage-import}.

This sufficient condition is not necessary.  The ternary construction in
Theorem~\ref{thm:rates} is not certified by the earlier two-stage criterion,
yet every one of its first-block pairs has finite extinction depth, with some
surviving beyond the old cutoff.  Thus the new criterion certifies a valid
correcting code that the earlier local test leaves unresolved.

We therefore continue every unresolved first-block pair through all legal
later block extensions on both sides.  A pair remains \emph{live} while a
possible collision is still unresolved and becomes \emph{extinct} once no
such ambiguity survives.  Extinction is irreversible: later block extensions
cannot recreate a separated pair.  If every first-block pair becomes extinct
at one common finite horizon, then a longer collision would have survived to
that horizon, while a shorter collision could be extended to it by appending
the same legal blocks on both sides.  Hence the same finite certificate proves
that \(C_n(P)\) is correcting for every \(n\).  The first separating horizon
is the \emph{extinction depth}.  It also measures how far a verifier must
explore before certifying the block set, although the total running time still
depends on the number of reachable configurations.

The criterion produces sharper fixed-alphabet constructions and also feeds a
large-alphabet lifting argument.  The normalized rate records the exponential
growth of the number of codewords with their length, while the capacity is the
supremum of asymptotic rates achievable by correcting codes.  Write
\(C_0^{(q)}\) for this capacity in the radius-one channel over a \(q\)-symbol
alphabet.  The new block sets constructed here have rates above \(0.6777475\)
for \(q=3\) and \(0.6694926\) for \(q=5\).

On the upper-bound side, Langberg, Schwartz, and Yaakobi used a length-three
covering code to bound \(C_0^{(q)}\)~\cite{LSY17}.  We first refine that cover
at length six, obtaining strict numerical improvements for every \(q\ge3\).
A length-eighteen refinement uses an exact balanced cover on every two-symbol
subalphabet and further strengthens the resulting \(q\)-ary bound.  We then lift covers of the
all-distinct permutation component at an arbitrary
fixed window length.  Combining this lift with the general permutation-cover
estimate of Farnoud, Schwartz, and Bruck~\cite{FSB16} gives
\[
 \liminf_{q\to\infty}(1-C_0^{(q)})\ln q\ge\ln\varphi,
 \qquad \varphi=\frac{1+\sqrt5}{2};
\]
see Theorem~\ref{thm:covering-phi}.

For lower bounds, the balanced two-class construction of Chee et
al.~\cite{CKLNVZ16} has rate
\(R_{\mathrm{Chee}}(q)=\frac12\log_q(\lfloor q/2\rfloor\lceil q/2\rceil)
=1-\log_q2+O(1/(q^2\ln q))\).  The later construction of
Kova\v{c}evi\'c, Goyal, and Kiah improves this benchmark for \(q=3,4\), but not
for \(q>4\)~\cite{KGK25}; the exact fixed-alphabet comparisons appear in
Section~\ref{sec:constructions}.  For growing alphabets, our run-wise lifting
construction gives
\[
 C_0^{(q)}\ge 1-\log_q(1.82560995)+o(1/\ln q),
\]
reducing the first-order correction loss factor from \(2\) to a value below
\(1.82560995\); see Theorem~\ref{thm:runwise-main}.  Together, the upper and
lower bounds place the normalized correction loss between \(\ln\varphi\) and
\(\ln(1.82560995)\).

The second main theme is error detection.  Detection of a single adjacent
transposition was already studied in the classical edit-error
setting~\cite{AbdelGhaffar98}; here one channel use may contain any set of
pairwise disjoint adjacent transpositions.  Detection forbids one codeword from
being obtainable from another, whereas correction requires their two sets of
possible outputs to be disjoint.  Stable lifting inside each type class
constructs detecting codes over every alphabet and gives a finite-length lower
bound.  Together with the pairing upper bound of~\cite{PriorWork2026}, the
all-alphabet construction shows that \(\sqrt2\) is the optimal first-order
large-alphabet loss factor.  We also construct explicit weak-zigzag detecting
codes, which improve the fixed-alphabet lower bounds for \(q=3,4\).  On the
structural side, detection is directional, so we follow ordered pairs of
distinct first blocks through later block extensions.  Directed extinction of
every such pair gives an all-length criterion for block-concatenation detecting
codes.  It also certifies a detecting block set that the earlier local test
leaves unresolved.

Finally, the window and cancellation mechanisms are not specific to radius
one.  For every fixed \(r\ge1\), we define \(r\)-truncated liveness and prove
that extinction at one common finite horizon certifies correction for
\(\LPC(r)\); the directed version gives the corresponding detection criterion.
A bounded pending-input frontier also makes correction and finite extinction
decidable by a finite graph at every fixed radius.

\subsection*{Our contributions}

\begin{enumerate}[leftmargin=*]
\item \textbf{Extinction depth and a one-condition criterion.}
For every pair of distinct first blocks, we follow all legal later block
continuations on both sides and discard a candidate once it is separated.  We
prove that extinction of every pair at one common finite horizon certifies
that every \(C_n(P)\) is correcting.  The earlier two-stage criterion implies
such extinction within its local cutoff.  A ternary correcting block set
constructed here is not certified by that earlier test, but all of its
first-block pairs nevertheless become extinct, proving that the new criterion
is strictly stronger.

\item \textbf{Stronger fixed-alphabet correcting codes.}
Explicit ternary and five-ary correcting block sets achieve rates above
\(0.6777475\) and \(0.6694926\), respectively.  The ternary rate improves the
directly comparable \(0.538196\) benchmark, while the five-ary construction
improves the unaugmented-template rate \(0.663964\).

\item \textbf{Linear and quadratic extinction depths.}
The two-block family \(\{0^L,0^{L-1}1\}\) has exact depth \(L+2\), proving
that no finite bound can be independent of the block lengths.  A second
two-block family has exact quadratic depth \(2k^2+4k-1\) for every
\(3\le k\le60\).  These results show that the certification horizon can be
substantially larger than the blocks themselves.

\item \textbf{Exact finite-state classification and certificates.}
We classify every first-block pair into finite extinction, a reachable finite
collision, or an infinite nonclosing ambiguity.  This identifies finite
extinction with acyclicity and reduces exact all-length verification to
reachability in the same finite graph.  Write
\(L_P=\sum_{w\in P}|w|\) and \(\lmax=\max_{w\in P}|w|\) for the total and
maximum block lengths.  Exact extinction verification runs in
\(O(|P|^2q^4(L_P+\lmax)^2)\) time and
\(O(q^2(L_P+\lmax)^2)\) working memory.  The graph also gives ranking certificates,
characterizes infinite depth by ultimately periodic common-output streams,
and bounds every finite depth quadratically in the total block length and the
maximum block length.  When all blocks have one length \(\ell\), the bound
sharpens to \((4q^2+1)\ell\).  These bounds control the number of
prefix-extension rounds required by breadth-by-depth certification.

\item \textbf{Large-alphabet correction bounds.}
An explicit length-six covering code improves the classical length-three
upper bound for every \(q\ge3\).  Writing \(U_q\) for the number of balls in
this length-six cover, combining three such windows with an exact balanced
two-symbol cover gives
\(K_q(18;1)\le U_q^3-4\binom q2\).  More generally, lifting permutation covers
through all-distinct windows proves
\[
 \liminf_{q\to\infty}(1-C_0^{(q)})\ln q\ge\ln\varphi.
\]
Run-wise lifting gives the complementary constructive bound
\[
 \limsup_{q\to\infty}(1-C_0^{(q)})\ln q
 \le\ln(1.82560995).
\]
Thus the asymptotic correction loss factor lies between \(\varphi\) and
\(1.82560995\), improving the Chee factor \(2\).

\item \textbf{Error-detecting codes.}
Stable type lifting constructs detecting codes for every alphabet and every
word length, including a finite-length code of size at least
\(q^n/2^{\lfloor n/2\rfloor}\).  Combined with the pairing upper bound
of~\cite{PriorWork2026}, this determines the optimal first-order
large-alphabet loss factor \(\sqrt2\).  We also give explicit weak-zigzag
detecting codes, with rates at least \(0.736917768\) for \(q=3\) and
\(0.762880412\) for \(q=4\).  Directed extinction then gives an all-length
criterion for block-concatenation detecting codes and certifies a detecting
block set left unresolved by the earlier local test.  The two-block corridor
family has exact certification depth \(4k\), in contrast with its quadratic
correction depth over the proved correction range.

\item \textbf{Every fixed displacement radius.}
The window and cancellation arguments extend beyond adjacent transpositions.
For every \(r\ge1\), a common finite extinction horizon certifies
\(r\)-correction, and directed extinction certifies \(r\)-detection.  A
finite graph with \(O((q+2)^{4r}N_P^2)\) states decides correction and finite
extinction; see Theorems~\ref{thm:depthKr}
and~\ref{thm:finite-state-r}.

\end{enumerate}

Section~\ref{sec:prelim} fixes notation and recalls the local facts used
below.  Section~\ref{sec:criterion} develops the extinction criterion, its
fixed-alphabet constructions, an exact linear depth family, and the exact
quadratic corridor depths over the proved finite range.
Section~\ref{sec:automaton}
then constructs the finite configuration graph, gives the verification
algorithm, and derives the finite-state classification, pumping bounds,
ranking certificates, and infinite-collision characterization.
Section~\ref{sec:runwise} gives the covering and constructive large-alphabet
correction bounds.
Section~\ref{sec:detection} gives the large- and fixed-alphabet detection
constructions, directed extinction criterion, and exact corridor depth.
Section~\ref{sec:lpcr} extends the depth criteria and finite-state correction
test to every fixed displacement radius.  We then discuss open problems and
conclude.
\section{Preliminaries}
\label{sec:prelim}

This section fixes the channel model and the notation used throughout the
paper.  We first describe the allowed permutations and distinguish correction
from detection.  We then introduce prefix and block-concatenation notation,
and finally recall the local facts on which the extinction-depth criterion is
built.

\subsection{Channel model and coding notions}

Throughout, \(q\ge2\), \(\mathbb N=\{1,2,\ldots\}\), and, for \(n\ge1\),
\([n]=\{1,\ldots,n\}\); set \([0]=\emptyset\).  For integers \(i\le j\),
write \([i,j]=\{i,i+1,\ldots,j\}\).  Let
\(\Sigma_q=\{0,\ldots,q-1\}\), and write
\(\Sigma_q^*\) for the set of all finite words over \(\Sigma_q\).  The empty
word is \(\epsilon\), the length of a word \(w\) is \(|w|\), and juxtaposition
denotes concatenation.  Let \(S_n\) be the symmetric group on \([n]\).  We
identify \(\pi\in S_n\) with its one-line word
\(\pi(1)\cdots\pi(n)\) whenever a permutation is used as a channel input.
For \(\pi\in S_n\), put
\(\wt(\pi)=\max_{i\in[n]}|\pi(i)-i|\), and for a word
\(w=w_1\cdots w_n\) write \((\pi w)_i=w_{\pi(i)}\).  The radius-\(r\)
output ball is
\[
 B_r(w)=\{\pi w:\pi\in S_n,\ \wt(\pi)\le r\}.
\]
We write \(\B(w)=B_1(w)\).  At radius one, every admissible
permutation is a product of pairwise disjoint adjacent transpositions.  A
\emph{swap site} is an index \(i\in\{1,\ldots,n-1\}\), representing the
adjacent pair of positions \((i,i+1)\).  A \emph{swap set} is a set
\(S\subseteq\{1,\ldots,n-1\}\) containing no consecutive integers, and
\(\sigma_S\) exchanges the two symbols at every site in \(S\).  Every
\(\sigma_S\) is an involution.

For example, when \(n=6\), the swap set \(S=\{1,4\}\) swaps the pairs in
positions \((1,2)\) and \((4,5)\).  The prohibition on consecutive elements
of \(S\) ensures that no position participates in two swaps.

A code \(\cC\subseteq\Sigma_q^n\) is \emph{correcting} if
\(\B(c)\cap\B(c')=\emptyset\) for all distinct \(c,c'\in\cC\), and it is
\emph{detecting} if \(c'\notin\B(c)\) for all distinct \(c,c'\in\cC\).
Two distinct equal-length words are \emph{confusable} if their radius-one
balls intersect; an element of the intersection is a common output witnessing
the collision.
Thus correction requires the received word to identify at most one transmitted
codeword, whereas detection only requires that an effective channel error
cannot turn one codeword into another.
Let \(A_q(n;1)\) be the largest size of a correcting code and define
\[
 C_0^{(q)}=\limsup_{n\to\infty}\frac1n\log_q A_q(n;1),
\]
with \(C_0=C_0^{(2)}\).
The capacity is the largest asymptotic base-\(q\) growth rate of correcting
codes.  Consequently, an explicit family of correcting codes of rate \(R\)
gives the lower bound \(C_0^{(q)}\ge R\).

A set \(\mathcal D\subseteq\Sigma_q^n\) is a \emph{radius-one covering code}
if
\[
 \bigcup_{d\in\mathcal D}\B(d)=\Sigma_q^n,
\]
and \(K_q(n;1)\) denotes the minimum possible size of such a set.  Since every
allowed permutation is an involution, two words in the same ball \(\B(d)\)
have the common output \(d\) and are confusable.  Hence every correcting code
contains at most one word from each ball of a covering code, and
\[
 A_q(n;1)\le K_q(n;1).
\]
This elementary covering bound was recorded for the limited permutation
channel in~\cite[Thm.~27]{LSY17}.

\subsection{Prefixes and block-concatenation codes}

For a word \(w\) and \(0\le d\le |w|\), let \(\prefx_d(w)\) denote its
length-\(d\) prefix, with
\(\prefx_0(w)=\epsilon\); for a set of words the same notation is applied
elementwise.  For a nonempty word \(w\in\Sigma_q^\ell\), let
\[
 \Tball(w)=\prefx_{\ell-1}(\B(w))
\]
be its truncated ball.  For equal-length words, let \(\dH\) denote Hamming
distance and define
\[
 \Dmin(u,v)=\min\{\dH(u',v'):u'\in\B(u),\ v'\in\B(v)\}.
\]
The final output position is omitted in \(\Tball(w)\) because, when \(w\) is
viewed as a window inside a longer word, a transposition crossing the end of
the window may still change that position.  The first \(\ell-1\) received
positions are already determined by the length-\(\ell\) transmitted prefix.
The quantity \(\Dmin(u,v)\) measures how closely the two channel balls can
approach one another.

For a finite nonempty block set
\(P\subseteq\Sigma_q^*\setminus\{\epsilon\}\), write \(P^*\) for its free
concatenation language,
\(C_n(P)=P^*\cap\Sigma_q^n\),
\(\lmax=\max_{w\in P}|w|\), and
\(L_P=\sum_{w\in P}|w|\).  The set \(P\) is \emph{prefix-free} if no block
is a prefix of a distinct block.  If \(X\subseteq\Sigma_q^*\), then
\(\Pref_d(X)\) denotes the set of length-\(d\) prefixes of words in \(X\)
having length at least \(d\), and \(\Pref(X)=\bigcup_{d\ge0}\Pref_d(X)\).  If
\(p_\ell=|P\cap\Sigma_q^\ell|\), then, to distinguish the several block sets
used below, denote the growth parameter of the profile by \(\lambda(P)\),
rather than the unindexed \(\lambda\) used in~\cite{PriorWork2026}.  It is the
unique positive solution of
\[
 \sum_{\ell\ge1} p_\ell\lambda(P)^{-\ell}=1,
 \qquad R(P)=\log_q\lambda(P).
\]
For prefix-free \(P\), this is the asymptotic word rate of \(C_n(P)\), exactly
as in~\cite[Sec.~2.1]{PriorWork2026}; we do not repeat that argument here.

Thus one finite block set describes codes at every total length.  Prefix
freeness makes the block parsing unambiguous, and the growth equation balances
the number of available blocks of each length against the cost of that length.

\subsection{Local facts behind finite-window certification}

The proofs below use three local facts from~\cite{PriorWork2026}.  We state them
in the present notation and refer to their original proofs.

\begin{lemma}[Window restriction]
\label{lem:window-import}
By~\cite[Lem.~16(a)]{PriorWork2026}, if \(c'\in\B(c)\) and \(1\le\ell\le|c|\), then
\[
 \prefx_{\ell-1}(c')\in\Tball(\prefx_\ell(c)).
\]
\end{lemma}

In this channel statement, \(c\) is the transmitted word and \(c'\) is one of
its possible received words.  The lemma says that the first \(\ell-1\) received symbols are
already determined by the first \(\ell\) transmitted symbols.  Indeed, at
radius one, a symbol occurring after position \(\ell\) cannot reach any of the
first \(\ell-1\) positions.  This is the basic reason finite windows can
certify arbitrarily long words.

\begin{lemma}[Truncated-ball test]
\label{lem:trunc-import}
By~\cite[Lem.~1]{PriorWork2026}, for distinct equal-length words \(u,v\), the same-length test of
\cite[Def.~10]{PriorWork2026} passes if and only if
\(\Tball(u)\cap\Tball(v)=\emptyset\).  In particular, this condition implies
\(\B(u)\cap\B(v)=\emptyset\), and \(\Dmin(u,v)\ge2\) implies truncated-ball
disjointness.
\end{lemma}

The truncated-ball test is the local separation step: disjoint truncated balls
show that the two transmitted prefixes can no longer produce the same
already-determined received prefix.

\begin{lemma}[Cancellation]
\label{lem:cancel-import}
By~\cite[Lem.~2]{PriorWork2026}, if \(\B(as)\cap\B(at)\ne\emptyset\), then
\(\B(s)\cap\B(t)\ne\emptyset\).  The directed first-block reduction used in
error detection is the corresponding step in the proof of
\cite[Thm.~17]{PriorWork2026}.  Hence a common block prefix may be cancelled
in both settings.
\end{lemma}

Cancellation removes common initial blocks before the two candidate
concatenations diverge.  Hence every possible collision can be reduced to a
pair beginning with two distinct first blocks.

The only new local observation needed for the one-condition theorem is that a
collision survives a common suffix.

\begin{lemma}[Common extension]
\label{lem:extend-collision}
If \(\B(x)\cap\B(y)\ne\emptyset\), then
\(\B(xp)\cap\B(yp)\ne\emptyset\) for every word \(p\).  Likewise,
\(y\in\B(x)\) implies \(yp\in\B(xp)\).
\end{lemma}
\begin{proof}
Use the realizing transpositions only inside the original prefix and use no
transposition at the join or inside the common suffix.
\end{proof}

For completeness, we restate the two-stage criterion used below.

\begin{theorem}[Two-stage local criterion; {\cite[Def.~11, Thm.~3, and Lem.~4]{PriorWork2026}}]
\label{thm:two-stage-import}
Let \(P\subseteq\Sigma_q^*\) be a finite set of nonempty blocks.  Assume that
the following conditions hold for every two distinct blocks \(a,b\in P\) with
\(|a|\le |b|\):
\begin{enumerate}[label=\textup{(T\arabic*)}]
\item If \(|a|=|b|\), then
\[
 \Tball(a)\cap\Tball(b)=\emptyset.
\]
\item If \(|a|<|b|\), put \(b_0=\prefx_{|a|}(b)\) and
\(d=|b|-|a|\).  Then \(\Dmin(a,b_0)\ge1\); moreover, if
\(\Dmin(a,b_0)=1\), then
\[
 \Tball(a\rho)\cap\Tball(b)=\emptyset
 \qquad\text{for every }\rho\in\Pref_d(P^*).
\]
\end{enumerate}
Then \(C_n(P)\) is correcting for \(\LPC(1)\) for every \(n\ge0\).  Moreover,
\(P\) is prefix-free.
\end{theorem}

This theorem gives an explicit finite certificate for all codes \(C_n(P)\)
simultaneously.  Prefix-freeness is a consequence of the unequal-length
condition, not an additional hypothesis: if \(a\) were a proper prefix of
\(b\), then \(b_0=a\) and \(\Dmin(a,b_0)=0\).  Prefix-freeness also ensures
that the profile rate \(R(P)\) is the actual word rate.
Proposition~\ref{prop:subsume}, together with the depth
criterion of Theorem~\ref{thm:depthK}, will give an alternative proof of its
soundness.  The extinction-depth criterion developed next continues the
same local analysis beyond the fixed two-stage cutoff by following every
unresolved pair through further legal block extensions.

\section{Extinction depth and code constructions}
\label{sec:criterion}

A failed local test does not necessarily exhibit a collision; it may only show
that the present window is too short.  We therefore follow the unresolved pair
through legal block continuations.  The basic monotonicity is one-sided: once
a pair is separated, no extension can make its shorter prefix ambiguous again.

\begin{definition}[Live pair]
\label{def:live}
For \(n\ge1\) and \(u,v\in\Sigma_q^n\), the unordered pair \(\{u,v\}\) is \emph{live} if
\(\Tball(u)\cap\Tball(v)\ne\emptyset\), and \emph{extinct} otherwise.
\end{definition}
Liveness is unresolved ambiguity, not the desired property: the certification
goal is to prove that every distinct first-block pair becomes extinct.

\begin{lemma}[Persistence]
\label{lem:persist}
If \(\B(c)\cap\B(c')\ne\emptyset\), then
\(\{\prefx_\ell(c),\prefx_\ell(c')\}\) is live for every
\(1\le\ell\le|c|\).  More generally, every positive-length prefix of a live
pair is live.
\end{lemma}
\begin{proof}
Choose a common output and apply Lemma~\ref{lem:window-import} to both
realizations.  The same argument applied to realizations witnessing a common
truncated output proves the prefix statement.
\end{proof}

Fix distinct first blocks \(a,b\in P\).  Although legal continuations are
generated by appending whole blocks, the index \(n\) below counts transmitted
symbols, not blocks.  Thus \(\Pref_n(aP^*)\) and \(\Pref_n(bP^*)\) may contain
prefixes that end inside their current blocks.  For \(n\ge1\), put
\[
 \Live_n(a,b)=
 \bigl\{\{u,v\}:u\in\Pref_n(aP^*),\ v\in\Pref_n(bP^*),
 \ \{u,v\}\text{ is live}\bigr\}.
\]
The \emph{extinction depth} of the first-block pair \((a,b)\) is
\[
 \ko(a,b)=\min\{n\ge1:\Live_n(a,b)=\emptyset\},
\]
with value \(\infty\) if the set never becomes empty.  When every first-block-pair depth
is finite, write
\[
 K_0(P)=\max_{\substack{a,b\in P\\a\ne b}}\ko(a,b).
\]
Thus \(\Live_n(a,b)\) is indexed by symbol length, and both \(\ko(a,b)\)
and \(K_0(P)\) are measured in symbols rather than in appended blocks.
Because liveness is symmetric in its two words,
\(\ko(a,b)=\ko(b,a)\); the order of the first blocks will matter only for detection.
\begin{example}[A first extinction]
For \(P=\{00,11\}\), the pair with first blocks \(00\) and \(11\) is live at depth one because
\(\Tball(0)=\Tball(1)=\{\epsilon\}\).  At depth two,
\(\Tball(00)=\{0\}\) and \(\Tball(11)=\{1\}\), so the pair is extinct and
\(\ko(00,11)=2\).  Thus extinction depth measures how long ambiguity survives,
not the length of a block in isolation.
\end{example}
Persistence makes extinction irreversible.  The next theorem explains why one
empty level is already a global correction certificate.

\begin{theorem}[Depth-\(K\) extinction criterion]
\label{thm:depthK}
Let \(P\subseteq\Sigma_q^*\) be a finite nonempty set of nonempty blocks.  If, for
some \(K\ge1\),
\[
 \tag{W}
 \Live_K(a,b)=\emptyset
 \qquad\text{for every distinct }a,b\in P,
\]
then \(C_n(P)\) is correcting for \(\LPC(1)\) for every \(n\ge0\).
\end{theorem}
Here \(n\) is the varying codeword length, whereas \(K\) is a single
certification horizon, measured in symbols, determined by \(P\); it is not a
number of blocks.  Intuitively, a collision of length
at least \(K\) leaves a live pair in its first length-\(K\) window, while a
shorter collision can be extended identically on both sides until it reaches
that window.  Hence no fixed ordering between \(n\) and \(K\) is assumed.
\begin{proof}
Suppose distinct \(c,c'\in C_n(P)\) have intersecting balls.  Fix block
factorizations and repeatedly cancel every common initial block by
Lemma~\ref{lem:cancel-import}.  This yields equal-length colliding words
\[
 x\in aP^*,\qquad y\in bP^*,\qquad |x|=|y|=:n',
\]
with \(a\ne b\).  Equality of the lengths is preserved because every
cancellation removes the same block from both sides; moreover \(n'>0\), since
cancelling all blocks would imply \(c=c'\).

If \(n'\ge K\), Lemma~\ref{lem:persist} places the length-\(K\) prefixes of
\(x\) and \(y\) in \(\Live_K(a,b)\), contradicting \textup{(W)}.  If
\(n'<K\), choose \(p\in P\) and \(m\ge1\) such that
\(n'+m|p|\ge K\).  Lemma~\ref{lem:extend-collision} preserves the collision
between \(xp^m\) and \(yp^m\), and Lemma~\ref{lem:persist} again places their
length-\(K\) prefixes in \(\Live_K(a,b)\), the same contradiction.
\end{proof}

The common-extension step is the reason no separate short-collision condition
appears in the theorem.  The depth criterion itself does not require
prefix-freeness; this property is
needed only when the growth parameter is identified with the word rate and in the
infinite-gap characterization of Section~\ref{sec:automaton}.

The new criterion genuinely extends the local test of~\cite{PriorWork2026}.  The
following proposition gives the precise implication and also explains the role
of the maximum-block-length cutoff
\(\lmax=\max_{p\in P}|p|\).

\begin{proposition}[The two-stage criterion implies extinction]
\label{prop:subsume}
If \(P\) satisfies the two-stage criterion of
Theorem~\ref{thm:two-stage-import}, then, for every distinct \(a,b\in P\),
\[
 \ko(a,b)\le\max\{|a|,|b|\}
 \qquad(a\ne b),
\]
and hence condition \textup{(W)} holds at \(K=\lmax\).
\end{proposition}
\begin{proof}
Assume \(|a|\le|b|\).  For equal lengths, the unique first-block pair at window
\(|a|\) is extinct by the same-length condition.  If \(|a|<|b|\) and
\(\Dmin(a,\prefx_{|a|}b)\ge2\), Lemma~\ref{lem:trunc-import} gives
extinction at window \(|a|\).  The value zero is excluded by the first
unequal-length condition.  In the remaining case, every length-\(|b|\) prefix
on the \(a\)-side has the form \(a\rho\) with
\(\rho\in\Pref_{|b|-|a|}(P^*)\), and the second-stage condition separates it
from \(b\).  Lemma~\ref{lem:persist} empties every later window.
\end{proof}

Combining Proposition~\ref{prop:subsume} with
Theorem~\ref{thm:depthK} gives an alternative proof of the conclusion of the
two-stage criterion, Theorem~\ref{thm:two-stage-import}.  In place of the
original direct induction on codeword length, this proof factors soundness as
\[
 \text{two-stage conditions}
 \ \Longrightarrow\ \text{extinction by }\lmax
 \ \Longrightarrow\ \text{correction at every length}.
\]

Proposition~\ref{prop:subsume} therefore gives
\[
 \text{two-stage certification}
 \ \Longrightarrow\ \text{finite extinction}
 \ \Longrightarrow\ C_n(P)\text{ is correcting for every }n.
\]
The first implication is strict by Theorem~\ref{thm:rates}; the second is strict by Example~\ref{ex:gapwitness}.

Reversal gives a useful symmetry of the channel and of every
block-concatenation construction.  It can therefore transfer a correcting or
detecting construction without repeating its verification.

\begin{lemma}[Reversal invariance]
\label{lem:reversal}
For a word \(w=w_1\cdots w_n\), write \(w^R=w_n\cdots w_1\), and put
\(P^R=\{p^R:p\in P\}\).  Then
\[
 \B(w^R)=\B(w)^R,
 \qquad
 C_n(P^R)=C_n(P)^R.
\]
Consequently, \(C_n(P)\) is correcting, respectively detecting, for every
\(n\) if and only if the same is true of \(C_n(P^R)\).
\end{lemma}

\begin{proof}
For a matching \(S\subseteq\{1,\ldots,n-1\}\), set
\(S^R=\{n-i:i\in S\}\).  This is again a matching, and
\(\sigma_{S^R}(w^R)=(\sigma_S(w))^R\), which proves the first identity.
Reversing a concatenation reverses both the factors and their order, proving
the second.  Reversal is a bijection and preserves both ball intersection and
directed ball membership, which gives the two coding equivalences.
\end{proof}

\label{sec:constructions}

We now apply the criterion to fixed-alphabet block sets.  The finite-state
procedure used to verify their extinction depths is developed independently
in Section~\ref{sec:automaton}; the present section records the resulting
mathematical constructions and rate bounds.

\subsection{Finite-alphabet block sets}

The closest recent fixed-alphabet comparison is the correcting construction of
Kova\v{c}evi\'c, Goyal, and Kiah~\cite{KGK25}.  For \(q=3\), after normalizing
logarithms to base three, their rate is
\(\log_3\xi_3=0.538196\ldots\), where \(\xi_3\) is the positive root of
\(x^6-3x^3-4x^2-4\); the Chee rate~\cite{CKLNVZ16} is \(0.315465\ldots\).  The ternary block
sets below improve both benchmarks.

The first role of extinction depth is to distinguish a genuinely bad block
set from one that merely survives the earlier local cutoff.  The
\(66\)-block ternary set \(P^{(3)}\), listed in
Appendix~\ref{app:blocks}, has profile
\[
 (p_3,p_4,p_6,p_8,p_{10})=(3,8,18,15,22).
\]
The two-stage criterion of Theorem~\ref{thm:two-stage-import} leaves twenty-seven stage-two
seed cases, arising from ten first-block pairs, unresolved.  None leads to a
collision: the reachable part of each associated configuration graph is
acyclic, and the deepest first-block pair dies at window \(24\).  All ten pairs
survive beyond \(\lmax=10\), so this block set gives the desired strict
separation between the old local cutoff and finite extinction.

The same idea can be used constructively.  When a candidate block leaves a
pair unresolved at depth \(\lmax\), we follow the pair rather than reject it.
The length-\(14\) extension used below is the \(265\)-block set
\(P^{(3)}_{14}\), whose size, extinction depth, growth parameter, and rate are shown below.

\begin{table}[!t]
\caption{The ternary base block set and its length-\(14\) extension.  Growth parameters
and rates are rounded to the displayed precision; rigorous lower bounds are
stated in Theorem~\ref{thm:rates}.}
\label{tab:ternary-rates}
\centering
\small
\begin{tabular}{@{}l r r r r@{}}
\toprule
block set & \(\lvert P\rvert\) & \(K_0(P)\) & \(\lambda(P)\) & \(R(P)\) \\
\midrule
\(P^{(3)}\)      & 66  & 24 & 2.0990359 & 0.6749225 \\
\(P^{(3)}_{14}\) & 265 & 36 & 2.1055608 & 0.6777476 \\
\bottomrule
\end{tabular}
\end{table}

\begin{theorem}[Ternary extinction-depth constructions]
\label{thm:rates}
The prefix-free block sets in Table~\ref{tab:ternary-rates} satisfy the
displayed extinction bounds.  In particular, \(P^{(3)}\) is not certified by
the two-stage criterion but satisfies the depth-\(24\) criterion, and
\[
 C_0^{(3)}\ge R(P^{(3)}_{14})>0.6777475.
\]
\end{theorem}
\begin{proof}
The exact finite-state traversal of Section~\ref{sec:automaton} constructs the
reachable part of every configuration graph and checks that it is acyclic,
with maximum depths \(24\) and \(36\), respectively.  For \(P^{(3)}\), it
also finds exactly ten first-block-pair depths above
\(\lmax=10\), proving strictness.  Theorem~\ref{thm:depthK} then proves
correction for both block sets.  The
growth equations have the exact rational brackets
\[
 \lambda(P^{(3)})\in(2.0990359,2.0990360),\qquad
 \lambda(P^{(3)}_{14})\in(2.1055607,2.1055608),
\]
and prefix-freeness identifies the growth parameters with the asymptotic word
rates.
\end{proof}

\subsection{A five-ary enhancement}
\label{sec:q5plus}

The ternary block sets show what deeper extinction can certify.  At \(q=5\),
the stronger of the Chee~\cite{CKLNVZ16} and
Kova\v{c}evi\'c--Goyal--Kiah~\cite{KGK25} benchmarks is the
Chee rate \(0.556641\ldots\) (the latter construction gives
\(0.503804\ldots\)).  A separate five-ary construction substantially
improves both through explicit symbolic templates.  Let \(P_5^+\) contain every assignment over
\(\Sigma_5\) of the following thirteen templates:
\[
\begin{array}{@{}c@{\qquad}l@{}}
\toprule
\text{symbol constraint} & \text{templates}\\
\midrule
a\in\Sigma_5
  & aaa\\
a\ne b
  & baaa,\quad aabbbb\\
a,b,c\ \text{pairwise distinct}
  & aabbcc,\quad abaccc,\quad abcccaa,\\
  & aabbbcaa,\quad abacbbbb,\quad abccaabb\\
a,b,c,d\ \text{pairwise distinct}
  & abcbdd,\quad abccadbb\\
a,b,c,d,e\ \text{pairwise distinct}
  & abcadcea,\quad abccdaea\\
\bottomrule
\end{array}
\]
The first eleven templates form a symbol-symmetric baseline.  The final two use
all five symbols and enlarge the length-eight layer while preserving every
two-stage local check; this is the source of the improvement.  Deleting them
leaves the unaugmented eleven-template family.  The enhanced profile is
\[
 (p_3,p_4,p_6,p_7,p_8)=(5,20,260,60,540).
\]

\begin{theorem}[Five-ary enhancement]
\label{thm:q5plus}
The block set \(P_5^+\) satisfies the two-stage local criterion of
Theorem~\ref{thm:two-stage-import}.  Consequently, \(C_n(P_5^+)\) is correcting for every
\(n\), and
\[
 C_0^{(5)}\ge \log_5\lambda_5^+>0.6694926,
\]
where \(\lambda_5^+\in(2.9373472,2.9373473)\) is the positive root of
\[
 5(\lambda_5^+)^{-3}+20(\lambda_5^+)^{-4}+260(\lambda_5^+)^{-6}
 +60(\lambda_5^+)^{-7}+540(\lambda_5^+)^{-8}=1.
\]
This improves the direct five-ary bound \(0.663964\) obtained from the
unaugmented templates.
\end{theorem}

\begin{proof}
An exhaustive evaluation of the finite two-stage conditions gives the following
exact verification counts:
\[
\begin{array}{@{}lr@{}}
\toprule
\text{test item} & \text{count}\\
\midrule
\text{blocks} & 885\\
\text{same-length unordered pairs} & 181{,}170\\
\text{same-length failures} & 0\\
\text{unequal-length pairs, shorter first} & 210{,}000\\
\text{first-stage pairs with }\Dmin\ge2 & 204{,}460\\
\text{pairs sent to stage two} & 5{,}540\\
\text{stage-two legal-prefix checks} & 544{,}500\\
\text{stage-two failures} & 0\\
\bottomrule
\end{array}
\]
Thus all conditions of Theorem~\ref{thm:two-stage-import} hold.  Its two
conclusions give all-length correction and prefix-freeness.  Exact sign
evaluation at the stated rational endpoints then brackets the unique positive
growth parameter and yields the rate bound.
\end{proof}

\subsection{Structural properties of extinction depth}
\label{sec:structure}

Before turning to the long corridor responsible for our deepest examples, it
is useful to separate two questions.  First, can finite extinction depth be
bounded independently of the block set?  Second, how rapidly can it grow with
the block lengths?  The next elementary family answers the first question
negatively.  Although it uses only the symbols \(0\) and \(1\), it is a
structural example over every alphabet of size at least two, rather than a
binary rate construction.

\begin{proposition}[An elementary linear family]
\label{prop:linear-depth}
Let \(q\ge2\), \(L\ge3\), and
\[
 P_L=\{0^L,0^{L-1}1\}\subseteq\Sigma_q^L.
\]
For its two distinct first blocks,
\[
 \ko(0^L,0^{L-1}1)=\ko(0^{L-1}1,0^L)=L+2.
\]
Consequently, finite extinction depth has no bound that is independent of the
block lengths.
\end{proposition}

\begin{proof}
Write \(a=0^L\) and \(b=0^{L-1}1\).  Up to depth \(L-1\), the two languages
beginning with \(a\) and \(b\) have the same all-zero prefix.  At depth \(L\), deleting the last
output coordinate still leaves the common word \(0^{L-1}\).  At depth \(L+1\),
the corresponding prefixes are \(0^{L+1}\) and \(0^{L-1}10\); swapping the final
\(10\) on the second side gives the common truncated output \(0^L\).  Thus the
pair is live through depth \(L+1\).

At depth \(L+2\), the corresponding prefixes are necessarily
\(0^{L+2}\) and \(0^{L-1}100\).  Under disjoint adjacent transpositions, the
unique \(1\) in the second word can move by at most one position, so it remains
among the first \(L+1\) output coordinates.  Every truncated output on that
side therefore contains a \(1\), whereas every output on the first side is
all zero.  Hence the pair is extinct at depth \(L+2\).  Symmetry gives the
reverse orientation.
\end{proof}

\paragraph{Quadratic depth on a finite corridor.}
\label{sec:lower}

Linear growth is only the first possible scale.  Two periodic parsings can
keep their relative phase moving through a product of residue classes.  The
following family has an exact quadratic depth formula throughout
\(3\le k\le60\).  The proof is a complete finite-state verification using
integer states and transitions.  The same formula for all \(k\ge3\) is stated
separately as a conjecture.

\begin{theorem}[Exact quadratic corridor depths]
\label{thm:corridor}
For \(3\le k\le60\), let
\[
\begin{aligned}
a_k&=00(10)^{k-3}1111, & |a_k|&=\ell_a=2k,\\
b_k&=(01)^{k-1}1110,   & |b_k|&=\ell_b=2k+2,
\end{aligned}
\]
and set \(P_k=\{a_k,b_k\}\).  Then \(P_k\) is prefix-free and
\[
\ko(a_k,b_k)=\ko(b_k,a_k)
=\lc(\ell_a,\ell_b)+\ell_a-1
=2k^2+4k-1.
\]
Consequently, \(P_k\) is correcting for every \(3\le k\le60\).
\end{theorem}

\begin{proof}
The exact finite-state method developed in Section~\ref{sec:automaton} gives a
breadth-by-depth traversal of the configuration dynamics in
Definition~\ref{def:configuration}.  For each ordered first-block pair it
starts from the unique depth-one configuration and, at every step, enumerates
all legal block-position continuations and all allowed pairs of swap decisions.  A
successor is retained exactly when it satisfies the equal-output edge condition in
Definition~\ref{def:configuration}.  Induction on the window length therefore
shows that the depth-\(n\) traversal level is nonempty exactly when a witness
path with \(n\) vertices beginning with \((a,b)\) exists.  By
Lemma~\ref{lem:paths-live}, this
is equivalent to \(\Live_n(a,b)\ne\varnothing\).

For every \(3\le k\le60\), the exhaustive computation finds nonempty levels
through \(n=2k^2+4k-2\) and an empty level at
\(n=2k^2+4k-1\), for both orders of the first blocks.  All states and transitions are
represented by integers, so the check is exact.  Since \(P_k\) has only these
two ordered first-block pairs, Theorem~\ref{thm:depthK} proves that every
\(C_n(P_k)\) is correcting.
\end{proof}

\begin{conjecture}[Uniform lcm corridor]
\label{conj:corridor-all-k}
The conclusion of Theorem~\ref{thm:corridor} holds for every \(k\ge3\).
Equivalently,
\[
\ko(a_k,b_k)=\lc(2k,2k+2)+2k-1=2k^2+4k-1
\qquad(k\ge3).
\]
If true, then \(L_{P_k}+\lmax=6k+4\), and this family satisfies
\[
 \ko(a_k,b_k)=\Theta\!\bigl((L_{P_k}+\lmax)^2\bigr),\qquad
 \frac{\ko(a_k,b_k)}{(L_{P_k}+\lmax)^2}\longrightarrow\frac1{18}.
\]
\end{conjecture}

The corridor was not chosen in isolation.  It emerged among the extremal
examples in an exhaustive sweep over binary two-block sets.  The sweep
suggests that the relevant length scale changes according to whether the two
block lengths are coprime.

\begin{conjecture}[Two-block depth law]
\label{conj:B2}
Let \(P=\{a,b\}\) be a prefix-free binary block set, with
\(\ell_a=|a|\le\ell_b=|b|\), and suppose that \(\ko(a,b)<\infty\).  If
\(\gcd(\ell_a,\ell_b)\ge2\), then
\[
 \ko(a,b)\le \lc(\ell_a,\ell_b)+\ell_a-1.
\]
If \(\gcd(\ell_a,\ell_b)=1\), then there is an absolute constant \(C\) such
that
\[
 \ko(a,b)\le C(\ell_a+\ell_b).
\]
\end{conjecture}

\begin{table}[!t]
\caption{Selected maxima from the exact two-block sweeps.  For each pair of
lengths, \(\max\ko\) is the largest finite depth among the indicated number
of symmetry classes; the displayed words form one extremal pair.}
\label{tab:b2sweep}
\centering
\scriptsize
\setlength{\tabcolsep}{3pt}
\begin{tabular}{@{}r r c r l r@{}}
\toprule
\(\ell_a\) & \(\ell_b\) & \(\gcd\) & \(\max\ko\) & witness \((a,b)\) & classes \\
\midrule
6  & 7  & 1 & 21 & \texttt{010100}, \texttt{0001011}         & 4032 \\
6  & 8  & 2 & 29 & \texttt{001111}, \texttt{01011110}        & 8064 \\
6  & 9  & 3 & 23 & \texttt{010000}, \texttt{001000001}       & 16128 \\
8  & 10 & 2 & 47 & \texttt{00101111}, \texttt{0101011110}    & 130560 \\
10 & 11 & 1 & 36 & \texttt{0001010111}, \texttt{00011101010} & 1047552 \\
\bottomrule
\end{tabular}
\end{table}

In every displayed noncoprime row, the maximum attains
\(\lc(\ell_a,\ell_b)+\ell_a-1\), whereas the coprime maxima remain on a
linear scale.  Fine--Wilf periodicity gives a natural explanation for this
split~\cite{FW65}.  A long live path forces the two block parsings to support
overlapping periodic stretches with periods \(\ell_a\) and \(\ell_b\); a
sufficiently long overlap then has period \(\gcd(\ell_a,\ell_b)\).  In the
coprime case this collapses to a constant stretch, which cannot sustain a
long effective swap relay.  With a nontrivial gcd, a nonconstant periodic
core can survive while the two parser phases move through their residue
classes, as they do in the corridor.  This is motivation for
Conjecture~\ref{conj:B2}, not a proof of it.

\section{Finite-state verification and the infinite gap}
\label{sec:automaton}

Our coding objective is
\[
 \B(u)\cap\B(v)=\emptyset
 \qquad
 \text{for every }n\text{ and all distinct }u,v\in C_n(P).
\]
Our finite-state method checks this condition by searching for the opposite
object: a pair
of distinct transmitted words and two legal channel actions that produce one
common output.  Such a pair is a collision witness.  If a complete witness is
found, the two balls intersect and the code is not correcting; if every
possible witness dies, the balls are disjoint.

Recall that swap site \(i\) is the adjacent pair of transmitted positions
\((i,i+1)\); it may lie inside one block or straddle a block boundary.  A legal
channel action is specified by a set of swap sites containing no consecutive
indices.  Exchanging the two symbols at every selected site produces one
possible received output.

Two fixed words already illustrate the search.  For
\(u=0000\) and \(v=1111\), every channel action leaves the respective word
unchanged, so the first settled output symbols are \(0\) and \(1\).  The
candidate witness dies immediately and
\(\B(0000)\cap\B(1111)=\emptyset\).  In contrast, for
\[
 u=1010,
 \qquad
 v=0101,
\]
selecting no swap sites on the first word and swap sites \(1,3\) on the
second gives
\[
 1010\longmapsto1010,
 \qquad
 0101\longmapsto1010.
\]
The selected swap sites are pairwise nonconsecutive, so this is a legal
collision witness.  The construction below records exactly such symbol
choices and swap decisions, but does so simultaneously for all legal
concatenations from \(P\).

For a fixed transmitted word \(w\), its channel ball \(\B(w)\) is a finite
language: its words are precisely the received outputs that can result from
\(w\).  We represent this language by a finite output automaton that accepts
exactly \(\B(w)\).  Hence two fixed-word balls intersect precisely when their
output automata accept a common word.

Block-concatenation codes require one representation covering words of
unbounded length.  For each possible first block
\(a\in P\), we therefore collect all finite balls above legal continuations in
the \emph{continuation-output language}
\[
 \mathsf{Out}_a(P)=\bigcup_{u\in aP^*}\B(u).
\]
The finite-state output automaton \(\mathcal A_P(a)\) recognizes this infinite
language.  Its cycles represent arbitrarily long block continuations, while
each state retains only the local information that determines the possible
future transitions and output labels.

The same idea could start from all of \(P^*\) by allowing every \(a\in P\) as
an initial block.  Its accepted language, however, would forget which first
block generated a given output.  Since correction requires comparing outputs
arising from distinct first blocks, we retain the separate languages
\(\mathsf{Out}_a(P)\).

To compare distinct first blocks \(a,b\in P\), we run
\(\mathcal A_P(a)\) and \(\mathcal A_P(b)\) synchronously and retain only
pairs of transitions that emit the same output symbol.  This synchronized
product is the finite configuration graph.  A path in it records two legal
transmitted prefixes and two legal channel actions that have produced the same
received prefix so far.  Thus a completed accepting path is a collision
witness, a dead end rules out that partial witness, and a reachable cycle
records a witness that can remain alive indefinitely.

The construction therefore has two layers.  We first build and analyze one
one-sided output automaton.  We then form the synchronized product for each
unordered pair of distinct first blocks.  We call the two hypothetical
transmitted words the left and right sides.  At each
step, the two automata extend transmitted prefixes in \(aP^*\) and \(bP^*\)
and label their transitions by newly settled received symbols.  The product
searches for a common output and thereby verifies the block set.

\subsection{One-sided output automata}
\label{sec:one-sided-output}

\begin{definition}[One-sided output automaton]
\label{def:one-sided-output}
Fix a finite set \(P\subseteq\Sigma_q^*\setminus\{\epsilon\}\) and a first
block \(c\in P\).  The automaton \(\mathcal A_P(c)\) generates the outputs of
all transmitted concatenations in \(cP^*\).  Recall that
\[
 L_P=\sum_{w\in P}|w|,
 \qquad
 \lmax=\max_{w\in P}|w|,
\]
and set
\[
 N_P=L_P+\lmax.
\]

The \emph{block position} records only how far the transmitted prefix has
progressed through its current block.  It is one of the following:
\begin{enumerate}[label=\textup{(\roman*)}]
\item \((c,j)\), where \(1\le j<|c|\), if the first \(j\) symbols of the
forced first block \(c\) have been appended;
\item \(r_0\), if the transmitted prefix ends exactly at a block boundary;
\item a nonempty proper prefix \(s\) of a block in \(P\), if the symbols
appended since the last block boundary spell \(s\).
\end{enumerate}
This coordinate records membership in \(cP^*\) and determines which input
symbols may legally be appended next.  At
\((c,j)\), the next symbol is forced to be \(c_{j+1}\).  At \(r_0\), a new
block begins; after a proper prefix \(s\), the next symbol must extend \(s\)
toward a block of \(P\).  The new position is \(r_0\) when a block is
completed, and otherwise it is the resulting proper prefix.  If both
descriptions apply, both continuations are retained, so the graph itself does
not require prefix-freeness.

After \(j\) transmitted symbols have been appended on one side, the first
\(j-1\) output positions can already be settled, but one end symbol may still
be affected by a swap at site \(j\).  We record this unresolved end as
\((\gamma,\beta)\).  Here \(\gamma\in\Sigma_q\) is the one held symbol, and
\(\beta\in\{0,1\}\) is the \emph{previous-swap flag}:
\[
 \beta=1
 \quad\Longleftrightarrow\quad
 \text{swap site }j-1\text{ was selected}.
\]
Thus \(\beta=1\) forbids selecting swap site \(j\), because the two swaps would
overlap at position \(j\); \(\beta=0\) leaves swap site \(j\) available.  The
flag records the immediately preceding decision in the candidate channel
action and enforces the nonoverlap constraint.

A \emph{one-sided state} of \(\mathcal A_P(c)\) is
\[
 \sigma=(\rho,\gamma,\beta).
\]
Here \(\rho\) is one of the block positions defined above.
One state represents all partial runs having the same block position and
unresolved end, because these data determine exactly the same future
transitions and output labels.  There are at most
\[
 (|c|-1)+1+\sum_{w\in P}(|w|-1)
 \le N_P
\]
block positions and \(2q\) possibilities for \((\gamma,\beta)\).  Thus its
state set \(Q_P(c)\) satisfies
\[
 |Q_P(c)|\le 2qN_P.
\]

Conceptually the automaton starts at the empty output \(\epsilon\).  No output
symbol is settled after only one transmitted symbol, so we suppress this
trivial start.  After the forced first symbol \(c_1\) has been appended, the
formal initial state is
\[
 \sigma_c^{(1)}=(\rho_c^{(1)},c_1,0),
 \qquad
 \rho_c^{(1)}=
 \begin{cases}
  (c,1),&|c|>1,\\
  r_0,&|c|=1.
 \end{cases}
\]

A transition appends one legal next transmitted symbol \(x\) and chooses the
channel action at the newly encountered swap site.  If the current
unresolved end is \((\gamma,\beta)\), exactly one new output symbol becomes
settled:
\[
\begin{array}{@{}c|c|c|c@{}}
\toprule
\text{decision at the new swap site}&\text{allowed when}
 &\text{settled output}&\text{new record}\\
\midrule
\text{do not select it}&\text{always}&\gamma&(x,0)\\
\text{select it for a swap}&\beta=0&x&(\gamma,1)\\
\bottomrule
\end{array}
\]
In the second row, the newly appended symbol moves one output position to the
left, while \(\gamma\) remains unresolved at the right end.  The settled output
is the transition label; the appended transmitted symbol and swap decision
specify how the transition is realized.

If the block position is \(r_0\), the transmitted label is a complete
concatenation and the automaton may terminate by emitting its held symbol
\(\gamma\).  Therefore the words accepted by \(\mathcal A_P(c)\) are exactly
the words in \(\mathsf{Out}_c(P)\).  We suppress the terminal edges and final
sinks in what follows.
\end{definition}

The size of one automaton is already polynomial in the explicit description
of \(P\).  There are at most \(N_P\) block positions,
\(q\) choices for the held symbol, and two choices for the previous-swap flag,
so
\[
 |Q_P(c)|\le 2qN_P.
\]
From one state, a transition chooses the next transmitted symbol, at most two
legal block-position successors, and whether the new swap site is
selected for a swap.  Hence there are at most \(4q\) candidate outgoing
transitions.  The entire one-sided automaton can therefore be generated in
\(O(q^2N_P)\) time; storing its states takes \(O(qN_P)\) space, and storing all of
its transitions takes \(O(q^2N_P)\) space.  This finite automaton replaces the
explicit enumeration of output prefixes and exponentially many swap histories.

The connection with extinction depth can already be stated on the one-sided
automata.  For \(n\ge1\), let
\[
 \mathsf{Out}_n(c)=
 \bigcup_{u\in\Pref_n(cP^*)}\prefx_{n-1}(\B(u)).
\]
This is the set of length-\(n-1\) output labels produced by paths of
\(\mathcal A_P(c)\) after \(n\) transmitted symbols have been appended.  A
single automaton determines \(\mathsf{Out}_n(c)\), but extinction is inherently
pairwise:
\[
 \Live_n(a,b)\ne\emptyset
 \quad\Longleftrightarrow\quad
 \mathsf{Out}_n(a)\cap\mathsf{Out}_n(b)\ne\emptyset.
\]
Consequently,
\[
 \ko(a,b)=
 \min\bigl\{n\ge1:
 \mathsf{Out}_n(a)\cap\mathsf{Out}_n(b)=\emptyset\bigr\},
\]
with value \(\infty\) if no such \(n\) exists.  The role of the synchronized
product is to test all these intersections without constructing the sets
\(\mathsf{Out}_n(c)\) explicitly.

\begin{example}[Four fixed words]
\label{ex:four-word-automata}
For a fixed transmitted word, the same construction gives a finite automaton
whose accepted language is its channel ball.  In particular,
\[
\begin{aligned}
 \B(0000)&=\{0000\},&
 \B(1111)&=\{1111\},\\
 \B(1010)&=\{1010,0110,1100,1001,0101\},\\
 \B(0101)&=\{0101,1001,0011,0110,1010\}.
\end{aligned}
\]
The first transitions from the initial states of the automata for \(0000\)
and \(1111\) have labels \(0\) and \(1\), respectively, so no synchronized
transition is possible.
Their languages are disjoint.  By contrast, the other two automata both have
the accepting path
\[
 \epsilon\xrightarrow{1}1\xrightarrow{0}10
 \xrightarrow{1}101\xrightarrow{0}1010.
\]
For the transmitted word \(1010\), this path uses no swaps.  For \(0101\), it
selects the two nonconsecutive swap sites \(1\) and \(3\).  Thus the
same output label \(1010\) is accepted on both sides.
\end{example}

\subsection{The synchronized product graph}
\label{sec:config-graph}

One one-sided automaton answers which received words can arise from one fixed
first block, but correction asks whether the same received word can arise from
two distinct first blocks.  Comparing the accepted words explicitly would
require enumerating exponentially many output prefixes.  The synchronized
product avoids this enumeration: it remembers only the current state of each
automaton and discards a pair of transitions as soon as their output labels
differ.  Reachability of an accepting pair then decides language intersection,
while cycles and longest paths measure indefinite and finite liveness.

\begin{definition}[Configuration graph]
\label{def:configuration}
Fix distinct first blocks \(a,b\in P\).  Their configuration graph is the
synchronized product
\[
 \mathcal G_P(a,b)=\mathcal A_P(a)\otimes\mathcal A_P(b).
\]
A vertex, called a \emph{configuration}, is simply a pair
\[
 \kappa=(\sigma_L,\sigma_R)
 \qquad
 \bigl(\sigma_L\in Q_P(a),\ \sigma_R\in Q_P(b)\bigr).
\]
We use \(\rho_L,\gamma_L,\beta_L\) for the three coordinates of \(\sigma_L\),
and similarly on the right.  The initial configuration is
\[
 \kappa_1=(\sigma_a^{(1)},\sigma_b^{(1)}).
\]
There is a product edge
\[
 (\sigma_L,\sigma_R)\longrightarrow(\sigma_L',\sigma_R')
\]
precisely when the left and right one-sided automata have transitions to
\(\sigma_L'\) and \(\sigma_R'\) with the same output label.  We label the
product edge by this common output symbol and annotate it by the two appended
transmitted symbols.

Writing \(V(\mathcal G_P(a,b))\) for the product vertex set, the one-sided
state bound gives
\[
 |V(\mathcal G_P(a,b))|
 \le |Q_P(a)|\,|Q_P(b)|
 \le4q^2N_P^2.
\]
Each product state has at most \((4q)^2=16q^2\) candidate outgoing
transition pairs before unequal output labels are discarded.  Thus the square
in the state bound is the natural price of language intersection; the gain is
that the bound is independent of the lengths of the concatenations and that
only reachable state pairs need be generated.
A configuration is \emph{co-complete} if \(\rho_L=\rho_R=r_0\).  It is
\emph{accepting} if it is co-complete and \(\gamma_L=\gamma_R\).  Indeed, in
that case both one-sided automata can take terminal transitions with the same
final output label, so they accept one common received word.
\end{definition}

\paragraph{Relation to the earlier verifier.}
The exact verifier of~\cite{PriorWork2026} also uses a product automaton, but
starts from a different one-sided machine.  Each of its two parsers reads the
same received word and reconstructs a possible transmitted concatenation.  A
parser step may reconstruct zero, one, or two symbols, so the product state
adds a one-symbol alignment buffer, denoted \(\mathrm{lead}\).  It also adds a
flag \(\mathrm{div}\) recording whether the two reconstructed words have
already differed.  In schematic form, its states are
\[
 (\text{parser state}_L,\text{parser state}_R,
   \mathrm{lead},\mathrm{div}).
\]
Those additions let one global automaton examine all pairs of codewords and
recognize a genuine collision between two distinct words.  In the present
graph, the first blocks \(a\ne b\) are fixed, so distinctness is automatic, and
both sides settle one output symbol per step, so no alignment buffer is
needed.  The simpler product \(\mathcal A_P(a)\otimes\mathcal A_P(b)\) is
therefore sufficient and, unlike the earlier orientation, its path length
records extinction depth directly.

\begin{definition}[Witness paths and dead ends]
\label{def:witness-path}
For \(n\ge1\) and
\((u,v)\in\Pref_n(aP^*)\times\Pref_n(bP^*)\), a \emph{witness path for
\((u,v)\)} is a path \(\kappa_1,\ldots,\kappa_n\) in
\(\mathcal G_P(a,b)\) whose \(j\)-th edge is annotated by
\((u_{j+1},v_{j+1})\) for \(1\le j<n\).  The initial vertex supplies
\((u_1,v_1)=(a_1,b_1)\).  Hence the labels spell the two transmitted prefixes,
while the transition choices specify one legal swap pattern on each side.

After \(j\) vertices, the two hypothetical channel outputs agree in their
first \(j-1\) positions; only \(\gamma_L\) and \(\gamma_R\) remain unresolved.
A nonaccepting configuration with no outgoing edge is a \emph{dead end}: no
continuation can preserve output agreement.  At an accepting configuration,
settling the two final symbols produces equal complete outputs and therefore
an actual intersection of the two channel balls.  Merely reaching a vertex or
exhausting one choice of swaps is not acceptance.
\end{definition}

The two fixed-word examples above now have exact graph interpretations.  For
\(0000\) versus \(1111\), every attempted first edge would settle \(0\) on
the left and \(1\) on the right, so there is no outgoing edge and the attempt
dies.  For \(1010\) versus \(0101\), the left decisions select no swap sites,
the right decisions select swap sites \(1,3\), and every newly settled output
symbol agrees.  After four vertices both words are complete and the two final
unresolved symbols are \(0\), so the path is accepting.  This is how the graph
distinguishes disjoint balls from intersecting balls.

Only the first blocks \(a\) and \(b\) are fixed in \(\mathcal G_P(a,b)\).
Every later block may be any member of \(P\), so the graph depends on the whole
block set.  To verify correction we inspect all \(\binom{|P|}{2}\) unordered
first-block pairs.  This is complete: after cancelling all common initial
blocks of a hypothetical collision, the remaining words begin with distinct
blocks \(a,b\in P\).

\begin{lemma}[Paths decide liveness]
\label{lem:paths-live}
The pair \(\{u,v\}\) is live if and only if \((u,v)\) admits a witness path
with \(n\) vertices, equivalently \(n-1\) edges.
\end{lemma}

\begin{proof}
A witness path determines nonconsecutive swap sets, and induction over its
edges shows that the two outputs agree in their first \(n-1\) positions.
Conversely, extract the swap decisions from two channel actions that agree in
those positions.  Nonconsecutiveness gives the update rule for
\(\beta_L,\beta_R\), equality of the settled outputs gives the edge condition,
and the legal transmitted prefixes give the block-position transitions.
\end{proof}

The lemma explains both the purpose of the graph and the role of extinction.
The graph enumerates, in compressed form, every pair of transmitted prefixes
and channel actions that could still lead to a common output.  If no path with
\(K-1\) edges exists, every collision attempt has died by depth \(K\), so the
live set at depth \(K\) is empty.  A path with \(r\) edges contains \(r+1\)
configurations and represents a live pair of length \(r+1\).

Two small block sets illustrate the remaining distinction.  For
\(P=\{00,11\}\), the initial configuration for \((00,11)\) has no outgoing
edge, because every possible channel decision settles unequal symbols.  For
\(P=\{01,1000\}\), the witness in Example~\ref{ex:gapwitness} reaches a
two-edge cycle.  Repeating the cycle preserves partial agreement for
arbitrarily long prefixes.  A cycle therefore proves infinite liveness, but
it is not accepting unless a path also reaches a co-complete configuration
with equal unresolved symbols.

\begin{proposition}[Acyclicity decides extinction]
\label{prop:acyclicity-extinction}
Let \(\mathcal G_P^{\mathrm{reach}}(a,b)\) be the subgraph reachable from
\(\kappa_1\).  Then
\[
 \ko(a,b)<\infty
 \quad\Longleftrightarrow\quad
 \mathcal G_P^{\mathrm{reach}}(a,b)\text{ is acyclic}.
\]
If this reachable graph is acyclic and \(d\) is the largest number of edges in
a path starting at \(\kappa_1\), then
\[
 \ko(a,b)=d+2.
\]
\end{proposition}

\begin{proof}
Suppose first that a directed cycle is reachable.  Follow a path from
\(\kappa_1\) to the cycle and then traverse the cycle arbitrarily many times.
Lemma~\ref{lem:paths-live} gives live pairs at unbounded depths.  If some level
were empty, persistence would force every later level to be empty as well, a
contradiction.  Hence \(\ko(a,b)=\infty\).

Conversely, if the reachable graph is acyclic, every path from \(\kappa_1\)
has at most \(d\) edges.  A longest path gives a live pair at depth \(d+1\),
whereas no path, and therefore no live pair, exists at depth \(d+2\).  Thus the
first empty level is exactly \(d+2\).
\end{proof}

\paragraph{Verification algorithm.}
To verify the whole block set, process every unordered pair
\(\{a,b\}\subseteq P\) with \(a\ne b\); there are at most
\(\binom{|P|}{2}\) such pairs.  For each pair, start at \(\kappa_1\) in
\(\mathcal G_P(a,b)\), generate legal outgoing edges on demand, and perform a
depth-first search.  Mark each reachable configuration white before it is
visited, gray while it is on the active recursion stack, and black after all
of its outgoing edges have been examined.  An edge whose head is gray is a
back edge and exhibits a reachable directed cycle.  Proposition~\ref{prop:acyclicity-extinction}
then gives \(\ko(a,b)=\infty\), so the finite-extinction criterion does not
certify \(P\).

If the search finishes without an edge to a gray configuration, the reachable
graph is acyclic.  Its DFS finishing order places every successor before its
predecessor, so a single dynamic-programming pass computes
\[
 h(\kappa)=
 \begin{cases}
  0,&\kappa\text{ has no outgoing edge},\\
  1+\displaystyle\max_{\kappa\to\kappa'}h(\kappa'),&\text{otherwise}.
 \end{cases}
\]
Thus \(h(\kappa_1)\) is the maximum number of edges in a path from the initial
configuration and
\[
 \ko(a,b)=h(\kappa_1)+2.
\]
If every first-block pair is acyclic, set
\[
 K_0(P)=\max_{\{a,b\}\subseteq P}\ko(a,b).
\]
Theorem~\ref{thm:depthK} then certifies that every \(C_n(P)\) is correcting.
A cycle for one pair has a deliberately weaker meaning: it proves that this
finite-extinction certificate fails, but it does not by itself prove a finite
codeword collision or invalidate the code.

\begin{proposition}[Complexity of extinction verification]
\label{prop:verification-complexity}
Put \(M=|P|\).  A one-sided automaton
\(\mathcal A_P(c)\) has \(O(qN_P)\) states and \(O(q^2N_P)\) candidate
transitions.  For one first-block pair, its reachable synchronized product can
be generated, tested for a directed cycle, and, when acyclic, used to compute
the exact extinction depth in
\[
 O(q^4N_P^2)\quad\text{time}
 \qquad\text{and}\qquad
 O(q^2N_P^2)\quad\text{working memory},
\]
when successors are generated on demand.  Processing the unordered distinct
first-block pairs sequentially therefore takes
\[
 O(M^2q^4N_P^2)\quad\text{time}
 \qquad\text{and}\qquad
 O(q^2N_P^2)\quad\text{working memory}.
\]
In particular, for fixed \(q\), extinction verification is polynomial in the
explicit size of the block set and is independent of the value of the depth
being certified.
\end{proposition}

\begin{proof}
Definition~\ref{def:configuration} gives at most \(4q^2N_P^2\)
configurations per first-block pair.  Equivalently, this is the product of two
one-sided state bounds of \(2qN_P\).  From one configuration there are \(q^2\)
choices of the next transmitted-symbol pair.  Each block-position transition
set has size at most two, giving at most four pairs of automaton successors,
and there are at most four pairs of swap decisions.  Hence there are at most \(16q^2\) candidate
outgoing transitions per configuration and at most \(64q^4N_P^2\) candidates
in the whole graph.

A depth-first traversal generates the reachable graph and detects a directed
cycle within these bounds.  Proposition~\ref{prop:acyclicity-extinction}
shows that this decides whether the extinction depth is finite.  In the
acyclic case, dynamic
programming along a reverse topological order computes a longest path and
hence the exact first empty level.  Storing colors and longest-path values for
the reachable configurations uses \(O(q^2N_P^2)\) memory when outgoing
transitions are regenerated rather than stored.  Finally, correction is
symmetric in the two first blocks, so there are at most
\(\binom{M}{2}\) first-block pairs.
\end{proof}

The graph introduced in Section~\ref{sec:config-graph} does more than verify a
finite cutoff.  It bounds every finite extinction depth, converts acyclicity
into a ranking certificate, and separates an infinite live cycle from a finite
collision.  For the exact-collision statements below we assume that \(P\) is
prefix-free, as in the exact finite-state formulation of~\cite{PriorWork2026}.
Recall from Definition~\ref{def:configuration} that a configuration has the
form
\[
 \kappa=(\sigma_L,\sigma_R),
\]
where each one-sided state is \(\sigma=(\rho,\gamma,\beta)\).  As before,
subscripts \(L,R\) indicate the coordinates of the left and right states.

For \(w\in\Sigma_q^n\), let
\[
 \comp(w)=(n_0(w),\ldots,n_{q-1}(w)),
\]
where \(n_i(w)\) counts the occurrences of \(i\), and let \(e_i\) be the
unit vector of \(\mathbb Z^q\) in the coordinate indexed by
\(i\in\Sigma_q\).
\begin{theorem}[Frontier identity and exact acceptance]
\label{thm:frontier-accept}
Let a witness path for \((u,v)\in\Sigma_q^n\times\Sigma_q^n\) end with unresolved
symbols \(\gamma_L,\gamma_R\).  Then
\[
 \comp(u)-\comp(v)=e_{\gamma_L}-e_{\gamma_R}.
\]
Moreover,
\[
 \B(u)\cap\B(v)\ne\emptyset
\]
if and only if \((u,v)\) admits a witness path with \(n\) vertices and
\(\gamma_L=\gamma_R\).
\end{theorem}
\begin{proof}
For the first assertion, induct on the path length and inspect the two-row
unresolved-end update in Definition~\ref{def:one-sided-output} on both sides.  If neither
side swaps, equality of the newly settled symbols forces the old unresolved
symbols to agree, and the new composition difference is the difference of the
two appended symbols.  If exactly one side swaps, its appended symbol equals
the unresolved symbol settled on the other side, so the identity is preserved.
If both sides swap, the appended symbols agree and the unresolved symbols are
unchanged.

A witness path already realizes agreement in the first \(n-1\) output positions.
Equality of the final unresolved symbols extends this agreement to position \(n\),
producing a common ball element.  Conversely, the swap decisions in any common
output give a witness path whose final unresolved symbols are equal.
\end{proof}

We can now separate the two questions that the local criterion deliberately
conflates: whether danger survives indefinitely, and whether it ever closes at
a pair of complete codewords.

Recall from Definition~\ref{def:configuration} that a configuration is
co-complete when \(\rho_L=\rho_R=r_0\), and accepting when it is co-complete and its
two unresolved symbols are equal.  Thus the labels on both sides are complete
concatenations and their balls intersect by Theorem~\ref{thm:frontier-accept}.
Here and below, the reachable configuration subgraph is the subgraph induced
by all configurations reachable from the initial state \(\kappa_1\).

\begin{theorem}[Exact finite-state classification]
\label{thm:state-classification}
Let \(P\) be prefix-free.  For each pair of distinct first blocks, exactly one of the
following outcomes holds:
\begin{enumerate}[label=\textup{(\alph*)}]
\item the reachable graph is acyclic; equivalently, the extinction depth is
finite;
\item an accepting configuration is reachable; equivalently, a finite
collision between two distinct complete concatenations beginning with the two
blocks exists;
\item a cycle is reachable but no accepting configuration is reachable; the
first-block pair is live at every depth but has no finite collision.
\end{enumerate}
Consequently, the extinction criterion is exactly acyclicity, whereas exact
correction verification is exactly non-reachability of accepting
configurations.
\end{theorem}
\begin{proof}
Proposition~\ref{prop:acyclicity-extinction} identifies the acyclic case with
finite extinction and the cyclic case with unbounded liveness.  At a
co-complete state,
Theorem~\ref{thm:frontier-accept} identifies acceptance with intersecting
balls of two complete concatenations.  Prefix-freeness makes the two labels
distinct because their first blocks differ.  Acceptance cannot occur in the
acyclic case: after a finite collision, append the same block repeatedly;
Lemma~\ref{lem:extend-collision} gives arbitrarily long colliding extensions,
so some configuration repeats.  The remaining cyclic case is therefore
precisely case~\textup{(c)}.
\end{proof}

The same finite-state graph therefore describes liveness, finite extinction,
exact collision, and the infinite gap.  For a non-prefix-free block set, the
state must additionally record whether the two parsings have already diverged,
so that two different parsings of the same word are not mistaken for two
codewords, as in the exact formulation of~\cite{PriorWork2026}.  In particular,
Theorem~\ref{thm:state-classification} explains why exact-collision
reachability is not an additional hypothesis of the extinction criterion.

\subsection{Pumping and ranking certificates}

The first consequence of finite memory is a universal depth bound: a witness path
that exceeds the state space must repeat a configuration and can therefore be
continued forever.

\begin{theorem}[Pumping]
\label{thm:pump}
For every first-block pair,
\[
 \ko(a,b)\le 4q^2(L_P+\lmax)^2+1
 \qquad\text{or}\qquad
 \ko(a,b)=\infty.
\]
Consequently, whether a first-block pair has finite extinction depth is decidable by
reachability in its finite configuration graph.
\end{theorem}
\begin{proof}
A witness path longer than the number of configurations repeats a configuration.  The
intervening legal label segment can then be repeated, producing live pairs at
arbitrarily large depths; persistence fills the omitted depths.  If no
configuration repeats, the path length is bounded by the number of configurations.
\end{proof}

The same path-and-cycle argument will prove the infinite-stream implication in Theorem~\ref{thm:gap}: arbitrarily long witness paths yield a reachable cycle, and repeating a cycle gives ultimately periodic labels and two infinite matchings with a common output.  For a finite-depth proof, acyclicity also has a direct mathematical witness: a rank that decreases at every step.

\begin{proposition}[Ranking certificates]
\label{prop:rank}
For a first-block pair, the following are equivalent:
\begin{enumerate}[label=\textup{(\roman*)}]
\item \(\ko(a,b)<\infty\);
\item the reachable configuration subgraph is acyclic;
\item there is \(\varphi:\{\text{reachable configurations}\}\to\mathbb N\) strictly decreasing along every transition.
\end{enumerate}
A sound certificate consists of a finite set of configurations that contains the
initial state and is closed under every legal outgoing transition, together
with ranks that decrease on all those transitions.  Once the legal successors
are generated, it suffices to check closure and the strict rank inequality once
for every legal transition.
\end{proposition}
\begin{proof}
(ii)\(\Leftrightarrow\)(iii) is the standard characterization of acyclicity in
a finite digraph: on a directed acyclic graph (DAG), rank each configuration
by one plus the length of its longest outgoing path; conversely a cycle admits no strictly decreasing
\(\varphi\) (infinite descent).
(ii)\(\Rightarrow\)(i): in an acyclic subgraph a path visits pairwise distinct configurations, so path lengths, and by Lemma~\ref{lem:paths-live} the depths of live pairs, are bounded by the number of reachable configurations.
(i)\(\Rightarrow\)(ii): a reachable cycle pumps witness paths to every sufficiently large length (proof of Theorem~\ref{thm:pump}), so \(\ko=\infty\).
\end{proof}

The general pumping bound counts block positions as well as unresolved-end records.
When all blocks have one length, synchronized block boundaries remove the
block-position coordinates and give a sharper linear cutoff.  This is the promised
upper counterpart to Proposition~\ref{prop:linear-depth}.

\begin{theorem}[Block-boundary pumping for one block length]
\label{thm:single-length-pumping}
Let \(P\subseteq\Sigma_q^\ell\), where \(\ell\ge1\), and let
\(a\ne b\in P\).  Then
\[
 \ko(a,b)\le(4q^2+1)\ell
 \qquad\text{or}\qquad
 \ko(a,b)=\infty.
\]
In particular, a finite extinction depth is at most \(17\ell\) when \(q=2\).
If every distinct first-block pair has finite extinction depth, then \(P\) satisfies
the depth-\((4q^2+1)\ell\) criterion.
\end{theorem}

\begin{proof}
Put \(K=(4q^2+1)\ell\), and suppose that
\(\Live_K(a,b)\ne\emptyset\).  Choose a live pair
\[
 u\in\Pref_K(aP^*),\qquad v\in\Pref_K(bP^*),
\]
and, by Lemma~\ref{lem:paths-live}, choose a witness path
\(\kappa_1,\ldots,\kappa_K\) on \((u,v)\).

After each complete block, the block position is \(r_0\).  Since every block
has length \(\ell\), the two block-position coordinates at the snapshots
\[
 \kappa_{\ell},\kappa_{2\ell},\ldots,
 \kappa_{(4q^2+1)\ell}
\]
are the same two block-boundary positions.  Only the two unresolved-end records
remain.  Each side has \(2q\) possible records,
\[
 (\gamma,\beta),
 \qquad \gamma\in\Sigma_q,\quad \beta\in\{0,1\},
\]
so there are at most \(4q^2\) pairs of records.

Among the \(4q^2+1\) snapshots, two configurations coincide.  The path segment
between them appends an integral number of blocks on both sides.  Because its
initial and final configurations agree, the same labels and swap choices may
be repeated indefinitely.  The block-boundary positions preserve legal block
continuation, and the previous-swap flags preserve the nonadjacency condition across every new
seam.  Hence witness paths beginning with \((a,b)\) exist at arbitrarily large depths.
Persistence then gives \(\Live_n(a,b)\ne\emptyset\) for every \(n\), so
\(\ko(a,b)=\infty\).  Thus liveness at depth \(K\) forces infinite extinction
depth; the contrapositive gives the finite bound, and
Theorem~\ref{thm:depthK} gives the final assertion.
\end{proof}

\subsection{The infinite gap}
\label{sec:gap}

A danger can survive at every finite depth without ever closing into a finite
collision.  Compactness of the automaton turns precisely this phenomenon into
a pair of right-infinite streams.  We state the infinite object immediately
before the theorem that uses it.

\begin{definition}[Legal infinite pair]
\label{def:legal-infinite}
Right-infinite words \(U,V\) are \emph{legal for \((a,b)\)} if every finite
prefix of \(U\) belongs to \(\Pref(aP^*)\) and every finite prefix of \(V\)
belongs to \(\Pref(bP^*)\).  A right-infinite word is \emph{ultimately
periodic} if it consists of a finite head followed by repetitions of a fixed
nonempty finite period.  A matching on \(\mathbb N\) is a subset
\(S\subseteq\mathbb N\) containing no two consecutive integers.  For such a matching \(S\), \(\sigma_S\)
simultaneously swaps coordinates \(i\) and \(i+1\) for every \(i\in S\).
\end{definition}

\begin{theorem}[Infinite-collision characterization]
\label{thm:gap}
Let \(P\) be prefix-free, \(a\ne b\in P\).
The following are equivalent:
\begin{enumerate}[label=\textup{(\roman*)}]
\item \(\ko(a,b)=\infty\);
\item the configuration subgraph reachable from the initial configuration of \((a,b)\) contains a cycle \textup{(}Section~\ref{sec:config-graph}\textup{)};
\item there exist ultimately periodic legal \(U,V\) and matchings \(S,T\) on
\(\mathbb N\) with \(\sigma_S(U)=\sigma_T(V)\).
\end{enumerate}
Consequently, for a finite prefix-free correcting block set, the depth-\(K\)
extinction condition fails for every finite \(K\) exactly when the reachable
graph for some distinct first-block pair contains a cycle.  Equivalently, for some distinct first blocks \(a,b\), there are legal
right-infinite \(U,V\) and matchings \(S,T\) with
\(\sigma_S(U)=\sigma_T(V)\), while no accepting configuration is reachable in
any distinct-first-block graph.
\end{theorem}

\begin{proof}
(iii)\(\Rightarrow\)(i) is a word-level implication.  For each \(n\), restrict
the two matchings to the first \(n-1\) swap sites.  Their outputs on the first
\(n-1\) coordinates are unchanged, so the length-\(n\) prefixes of \(U\) and
\(V\) form a live pair.

(i)\(\Rightarrow\)(ii) follows from Theorem~\ref{thm:pump}.  Conversely, a
reachable cycle supplies an ultimately periodic legal label pair and two
infinite matchings, proving (ii)\(\Rightarrow\)(iii).

For the consequence, Theorem~\ref{thm:state-classification} says that finite
extinction fails exactly when a first-block graph contains a cycle, whereas exact
correction fails exactly when an accepting state is reachable.  Thus the
certification gap consists of a cyclic distinct-first-block graph together with
non-reachability of accepting configurations in every distinct-first-block graph.
Prefix-freeness excludes \(U=V\) with distinct first blocks.
\end{proof}

The lasso construction in the preceding proof permits the two matchings to
be chosen ultimately periodic as well.  More precisely, after a finite
initial segment, the two input streams and the two indicator sequences of
the matchings repeat with one common period.  We call such a quadruple an
\emph{ultimately periodic common-output witness}.  Overlaying the two
matchings reveals the possible shapes of its periodic tail.

\begin{proposition}[Overlay trichotomy]
\label{prop:trichotomy}
Let \(P\) be prefix-free and correcting, and let \((U,V,S,T)\) be an
ultimately periodic common-output witness with distinct first blocks.  Regard
the swap sites in \(S\) and \(T\) as edges of the one-way path on
\(\mathbb N\).  The symmetric difference \(S\triangle T\) is a disjoint
union of maximal intervals whose edges alternate between the two matchings,
and \(U\) and \(V\) agree coordinatewise outside those intervals.  After a
finite initial segment, exactly one of the following occurs:
\begin{enumerate}[label=\textup{(\Alph*)}]
\item one alternating interval continues forever; on it the two matchings
are the two alternating brick walls and, if the interval begins at position
\(j\), \(x_t=U_{j+t}\), and \(y_t=V_{j+t}\), then
\[
 y_0=x_1,\qquad
 y_{2t}=x_{2t-2}\quad(t\ge1),\qquad
 y_{2t+1}=x_{2t+3}\quad(t\ge0);
\]
\item no alternating interval remains, so \(U=V\) eventually, but the two
block parsings never again share a block boundary;
\item a nonempty finite collection of alternating intervals repeats
periodically.
\end{enumerate}
In all three cases, the eventual period words on the \(U\)- and \(V\)-sides
have the same symbol multiplicities.
\end{proposition}

\begin{proof}
The union of two matchings has maximum degree two.  After common edges are
removed, every nontrivial component therefore alternates between an edge of
\(S\) and an edge of \(T\), and the two channel actions are identical outside
these components.  Equality of the received streams then forces \(U=V\)
there.  An infinite component must use the two alternating brick walls;
reading the common output successively from its first position gives the
displayed relations.

The eventual periodicity leaves three possibilities: one component is
infinite, the symmetric difference is eventually empty, or finite components
recur periodically.  In the second case a later common parser boundary would
close the witness into a finite accepting configuration, contradicting that
\(P\) is correcting.  Finally, each channel action moves a symbol across a
fixed cut by at most one position.  Hence the difference between the symbol
counts in equally long prefixes of \(U\) and \(V\) is bounded independently
of the prefix length.  Repeating the two period words \(m\) times makes this
difference \(m\) times their composition difference, up to a fixed head
term.  The difference can remain bounded for all \(m\) only when the two
period words have identical compositions.
\end{proof}

\begin{example}[A compact gap witness]
\label{ex:gapwitness}
The prefix-free block set \(P=\{01,1000\}\) is correcting, but
\(\ko(01,1000)=\infty\).  A complete traversal of the reachable
configuration graph finds no accepting configuration;
Theorem~\ref{thm:state-classification} therefore proves the correction claim.
A path-and-cycle witness in this no-acceptance graph has head
\((01010,10000)\) and period labels \((10,10)\).  At windows four and five,
\[
 \Tball(1000)\subseteq\Tball(0101),
 \qquad
 \Tball(10000)\subseteq\Tball(01010),
\]
and repeating the cycle continues this agreement indefinitely.  Along the cycle the
unresolved symbols remain different, with composition defect \(e_1-e_0\);
therefore Theorem~\ref{thm:frontier-accept} prevents this witness from closing
into a finite collision.  The displayed head and period make the witness directly checkable.
\end{example}

An exact exhaustive census over \(41{,}731\) prefix-free binary block sets
(pairs with block lengths at most \(7\), triples at most \(5\), and
quadruples at most \(4\)) found \(1{,}441\) instances of the certification
gap: \(504\) two-block, \(925\) three-block, and \(12\) four-block instances.
Every recorded cycle had period labels with equal symbol multiplicities, as
Proposition~\ref{prop:trichotomy} requires.  In all \(504\) two-block
instances the two period labels were cyclic rotations of one another, whereas
exactly \(16\) of the three-block instances had period labels that were not
cyclic rotations of one another.
These data motivate the classification problem in Section~\ref{sec:open};
they are not used in any proof.

Example~\ref{ex:gapwitness} therefore exposes a gap for an individual block
set, but not an asymptotic rate loss.  The two-stage criterion is rate-complete:
\cite[Thm.~7]{PriorWork2026} proves that the supremum of the rates it certifies
equals \(C_0^{(q)}\).  Since Proposition~\ref{prop:subsume} shows that every
two-stage certificate is also a finite-extinction certificate, the
extinction-depth criterion inherits the same rate-completeness.

\section{Large-alphabet correction bounds}
\label{sec:runwise}

This section bounds the first-order correction loss from both sides.  The
natural normalized loss is
\[
 \Delta_q=(1-C_0^{(q)})\ln q,
\]
or, equivalently, the multiplicative loss factor
\(q^{1-C_0^{(q)}}=e^{\Delta_q}\).  The two sides of the argument point in
opposite directions.  A small covering code gives an upper bound on the
correction capacity and hence a lower bound on \(\Delta_q\).  An explicit
correcting code gives a lower bound on the capacity and hence an upper bound
on \(\Delta_q\).  We first strengthen the covering bound, using an explicit
length-six cover for finite alphabets and permutation covers of arbitrary
fixed length for the asymptotic result.  We then improve the constructive
Chee loss factor by lifting a small-alphabet code run by run.

\subsection{Fixed-length covering upper bounds}
\label{sec:covering-upper}

Recall that \(K_q(n;1)\) is the minimum number of radius-one balls needed to
cover \(\Sigma_q^n\), whereas \(A_q(n;1)\) is the maximum size of a correcting
  code.  As shown in Section~\ref{sec:prelim},
\(A_q(n;1)\le K_q(n;1)\): a correcting code can contain at most one word from
each ball of a covering code.  Thus every explicit cover below translates
into an upper bound on \(C_0^{(q)}\).

When all symbols in a word are distinct, relabeling them by
\(1,\ldots,b\) identifies the word with a permutation, and the allowed channel
actions remain exactly the disjoint adjacent swaps.  A cover of this
permutation component can therefore be relabeled over each choice of distinct
alphabet symbols.

For \(q\ge b\), let \(S_b\) denote the symmetric group on \(b\) symbols and
write
\[
 (q)_b=q(q-1)\cdots(q-b+1)=b!\binom qb.
\]
Thus \((q)_b\) is the number of length-\(b\) words with distinct symbols.  Let
\(\tau_b\) be the minimum number of radius-one balls required to cover
\(S_b\), under the same disjoint-adjacent-swap channel.  Equivalently, after
identifying a permutation with a word containing each of \(b\) prescribed
symbols exactly once, \(\tau_b\) balls cover that entire all-distinct
component.

\begin{proposition}[Lifting an all-distinct permutation cover]
\label{prop:qary-cover-lift}
For every \(q\ge b\ge1\),
\[
 K_q(b;1)
 \le
 q^b-(q)_b+\tau_b\binom qb.
\]
Consequently,
\[
 C_0^{(q)}\le\frac1b\log_q K_q(b;1).
\]
\end{proposition}

\begin{proof}
There are \((q)_b\) words of length \(b\) with pairwise distinct symbols.
For each \(b\)-subset \(U\subseteq\Sigma_q\), relabel a fixed
\(\tau_b\)-ball cover of \(S_b\) by the symbols of \(U\).  This covers every
word whose symbol set is exactly \(U\), and doing so for all \(\binom qb\)
choices of \(U\) covers the all-distinct words.  Include each of the remaining
\(q^b-(q)_b\) words as an additional center.  The resulting radius-one cover
has the asserted size.

For the capacity statement, write \(n=mb+r\), where \(0\le r<b\).
Concatenating \(m\) length-\(b\) covers and using all \(q^r\) suffixes as
centers gives
\[
 K_q(n;1)\le K_q(b;1)^m q^r.
\]
Indeed, the internal swap sets from consecutive length-\(b\) segments remain disjoint and
therefore form one legal channel action.  Now use
\(A_q(n;1)\le K_q(n;1)\), divide the logarithm by \(n\), and let
\(n\to\infty\).
\end{proof}

The general lift has the correct leading behavior as \(q\to\infty\), but it
treats every repeated-symbol word as a separate center.  At length six we can
do substantially better for finite \(q\) by refining the product of the
standard length-three cover of Langberg, Schwartz, and
Yaakobi~\cite[Thm.~29]{LSY17}.  That cover has
\[
 T_q=q+2\binom q2+2\binom q3=\frac{q^3+2q}{3}
\]
centers.  Applying it independently to the first and last three coordinates
gives a length-six product cover with \(T_q^2\) centers.

The channel preserves the multiset of symbols in a word, so the cover can be
improved separately inside each multiset class.  We describe such a class by
its \emph{type}: the decreasing list of its nonzero symbol multiplicities.
For example, type \(2211\) consists of words using four distinct symbols, two
of them twice and two of them once.  After the actual symbols and their
multiplicities have been fixed, we call the resulting set of words a concrete
type class.  The following exact replacements reduce the product cover inside
five families of concrete type classes.

\begin{table}[t]
\centering
\caption{Exact length-six replacements inside fixed multiset types.}
\label{tab:length-six-types}
\begin{tabular}{@{}c@{\qquad}r@{\qquad}r@{}}
\toprule
type & product cover & replacement\\
\midrule
\(222\)    & 10 & 9\\
\(3111\)   & 16 & 15\\
\(2211\)   & 20 & 17\\
\(21111\)  & 40 & 32\\
\(111111\) & 80 & 60\\
\bottomrule
\end{tabular}
\end{table}

For each row, an explicit list of centers for one canonical concrete class was checked
by enumerating all thirteen legal swap sets at length six.  The verification
checks that the centers are distinct, that every generated word remains in
the declared type, and that every word of that type is covered.  Thus the
table is an exact finite certificate of existence, and relabeling transfers
the check to every concrete class of the same type.  No optimality of any
replacement size is claimed.  The center lists and solver-free verifier are
included in the accompanying reproducibility package~\cite{ReproPackage2026}.

\begin{theorem}[Explicit length-six covering bound]
\label{thm:length-six-cover}
For every \(q\ge2\), put
\[
 U_q=T_q^2-\binom q3-22\binom q4-40\binom q5-20\binom q6.
\]
Then
\[
 K_q(6;1)\le U_q,
 \qquad
 C_0^{(q)}\le\frac16\log_q U_q.
\]
Moreover,
\[
 U_q=\frac1{12}q^6+O(q^5),
 \qquad
 C_0^{(q)}
 \le
 1-\frac{\ln12}{6\ln q}
 +O\!\left(\frac1{q\ln q}\right).
\]
\end{theorem}

\begin{proof}
Start from the \(T_q^2\)-center product cover.  A concrete class of type
\(222\) saves one center, and there are \(\binom q3\) choices of its three
symbols.  For type \(3111\), choose four symbols and then the tripled symbol;
the saving is therefore \(4\binom q4\).  For type \(2211\), choose four
symbols and the two doubled symbols.  Since each class saves three centers,
the total saving is \(3\binom42\binom q4=18\binom q4\).  Similarly, type
\(21111\) saves eight centers for each choice of five symbols and of the
doubled symbol, giving \(8\cdot5\binom q5=40\binom q5\).  Finally, each
all-distinct class saves twenty centers, giving \(20\binom q6\).  Thus the
five total savings are
\[
 \binom q3,
 \quad 4\binom q4,
 \quad 18\binom q4,
 \quad 40\binom q5,
 \quad 20\binom q6
\]
centers.  Distinct concrete multiset classes are disjoint and are preserved by
the channel, so all replacements can be made simultaneously.  This gives \(U_q\).
Proposition~\ref{prop:qary-cover-lift}, or the same blockwise-concatenation
argument, gives the capacity bound.  Finally,
\(T_q^2=q^6/9+O(q^4)\) and
\(20\binom q6=q^6/36+O(q^5)\); all other savings have degree at most five.
Hence \(U_q=q^6/12+O(q^5)\), which yields the displayed asymptotic estimate.
\end{proof}

The length-six replacements use at least three distinct symbols.  Consequently,
on every fixed two-symbol subalphabet the underlying product of standard
length-three covers remains intact.  This permits a further refinement that
improves the unified \(q\)-ary bound.

\begin{theorem}[Unified length-eighteen covering bound]
\label{thm:length-eighteen-cover}
For every \(q\ge2\),
\[
 K_q(18;1)\le U_q^3-4\binom q2,
 \qquad
 C_0^{(q)}\le
 \frac1{18}\log_q\!\left(U_q^3-4\binom q2\right).
\]
\end{theorem}

\begin{proof}
Apply the \(U_q\)-center cover independently to three consecutive
length-six windows.  This gives a cover of \(\Sigma_q^{18}\) with
\(U_q^3\) centers.  Fix an unordered pair of alphabet symbols and divide the
eighteen coordinates into six consecutive triples.  None of the five
length-six replacements in Table~\ref{tab:length-six-types} affects words
supported on this pair.  The remaining standard product cover on these words
therefore has \(4^6\) centers, indexed by the six triple weights in
\(\{0,1,2,3\}\).

Exactly
\[
 [z^9](1+z+z^2+z^3)^6=580
\]
of these centers have total second-symbol weight nine.  After relabeling the
chosen pair, an exact balanced two-symbol cover gives \(576\) radius-one balls
covering the entire component in which each symbol occurs nine times.
Replacing the \(580\) product centers by these \(576\) centers saves four
centers for the chosen alphabet pair.  The
replacement balls remain inside that component because the channel preserves
symbol multiplicities.  Balanced components belonging to distinct unordered
pairs are disjoint, so all \(\binom q2\) replacements can be made
simultaneously.  This proves the stated cover size.  The certificate and its
solver-free verifier are included in the reproducibility
package~\cite{ReproPackage2026}; concatenating the resulting length-eighteen
cover proves the capacity bound.
\end{proof}

The improvement is already visible over small alphabets.  The old column in
Table~\ref{tab:length-six-numerics} is the product of two standard
length-three covers, and the new column is Theorem~\ref{thm:length-six-cover}.
Both rate columns are upper bounds on the correction capacity, so the smaller
number is the stronger bound.

\begin{table}[t]
\centering
\caption{Finite-alphabet correction-capacity upper bounds.}
\label{tab:length-six-numerics}
\begin{tabular}{@{}c@{\quad}c@{\quad}c@{\quad}c@{}}
\toprule
\(q\) & centers \(T_q^2\to U_q\) & old rate & new rate\\
\midrule
3  & \(121\to120\)       & 0.727553 & 0.726294\\
4  & \(576\to550\)       & 0.764160 & 0.758607\\
5  & \(2025\to1865\)     & 0.788404 & 0.779881\\
6  & \(5776\to5166\)     & 0.805676 & 0.795294\\
10 & \(115600\to96580\)  & 0.843826 & 0.830815\\
\bottomrule
\end{tabular}
\end{table}

We now let the permutation length grow.  Use the Fibonacci convention
\(F_0=0\), \(F_1=1\), and
\(F_{j+1}=F_j+F_{j-1}\) for \(j\ge1\).  A radius-one ball about any
permutation in \(S_b\) contains exactly \(F_{b+1}\) permutations, because a
legal swap set is a matching in the path on the \(b\) positions.  The
incidence system is regular, with universe size \(b!\) and ball size
\(F_{b+1}\).  Applying the general greedy set-cover estimate recorded by
Farnoud, Schwartz, and Bruck~\cite[Thm.~4]{FSB16} gives
\[
 \tau_b
 \le
 \frac{b!}{F_{b+1}}\bigl(1+\ln F_{b+1}\bigr).
\]

\begin{theorem}[Asymptotic covering loss]
\label{thm:covering-phi}
Let \(\varphi=(1+\sqrt5)/2\).  Then
\[
 \liminf_{q\to\infty}
 (1-C_0^{(q)})\ln q
 \ge \ln\varphi.
\]
\end{theorem}

\begin{proof}
Fix \(b\).  Proposition~\ref{prop:qary-cover-lift} gives
\[
 K_q(b;1)
 \le
 q^b\left(\frac{\tau_b}{b!}+O_b(q^{-1})\right),
\]
and hence
\[
 \liminf_{q\to\infty}(1-C_0^{(q)})\ln q
 \ge
 \frac1b\ln\frac{b!}{\tau_b}
 \ge
 \frac{\ln F_{b+1}-\ln(1+\ln F_{b+1})}{b}.
\]
This holds for every fixed \(b\).  Binet's formula gives
\(b^{-1}\ln F_{b+1}\to\ln\varphi\), whereas
\(b^{-1}\ln(1+\ln F_{b+1})\to0\).  Letting \(b\to\infty\) proves the claim.
\end{proof}

The constant is sharp for this all-distinct fixed-length ball-cover scheme.
Indeed, the volume bound gives \(\tau_b\ge b!/F_{b+1}\), so the coefficient
\(b^{-1}\ln(b!/\tau_b)\) obtainable from a length-\(b\) permutation cover is
at most \(b^{-1}\ln F_{b+1}\), which tends to \(\ln\varphi\).  The general
set-cover estimate loses only the subexponential factor
\(1+\ln F_{b+1}\) and therefore attains the same limiting coefficient.

\subsection{Run-wise lifting beyond the Chee construction}
\label{sec:runwise-lifting}

The Chee construction partitions the alphabet into two nearly equal classes
and fixes an alternating class pattern.  Its rate
\[
 \frac12\log_q\!\bigl(\lfloor q/2\rfloor\lceil q/2\rceil\bigr)
 =1-\log_q2+O\!\left(\frac1{q^2\ln q}\right)
\]
has multiplicative loss factor \(2\).  Alternation makes every adjacent swap
visible after projecting symbols to their classes.  To improve this factor we
use a richer projected code, whose words may contain constant-color runs.  A
swap within such a run is invisible to the projection, so the symbols inside
each run must themselves form a correcting code.  The following lifting
theorem separates these two tasks: a small-alphabet \emph{skeleton} protects
the sequence of colors, and an inner code protects the symbols in each
maximal constant-color run.

Let \(1\le m\le q\), and let
\[
  \Sigma_q=A_0\sqcup A_1\sqcup\cdots\sqcup A_{m-1}
\]
be a partition of the $q$-ary alphabet into nonempty color classes.  Let
\[
  \chi:\Sigma_q\longrightarrow\{0,1,\ldots,m-1\}
\]
be the induced color map, extended coordinatewise to words.
For a word \(x\in\Sigma_q^n\), the color word \(\chi(x)\) is its skeleton.

\begin{theorem}[Run-wise lifting]
\label{thm:runwise-lifting}
Let $\cD_n\subseteq\{0,1,\ldots,m-1\}^n$ be a correcting code for
$\LPC(1)$.  For every color $i$ and every run length $s\ge1$, let
\[
  E_{i,s}\subseteq A_i^s
\]
be a correcting code for $\LPC(1)$.

For $c\in \cD_n$, write its decomposition into maximal constant runs as
\[
  c=i_1^{s_1}i_2^{s_2}\cdots i_t^{s_t},
  \qquad i_j\ne i_{j+1}.
\]
Define $L(c)$ to contain all words
\[
  x=x^{(1)}x^{(2)}\cdots x^{(t)},
  \qquad x^{(j)}\in E_{i_j,s_j},
\]
and put
\[
  L(\cD_n)=\bigcup_{c\in \cD_n}L(c).
\]
Then $L(\cD_n)$ is a correcting code for $\LPC(1)$.
\end{theorem}

Every word in \(L(c)\) has skeleton \(c\).  Hence the sets \(L(c)\),
\(c\in\cD_n\), are pairwise disjoint, and the independent choices in the runs
give
\[
 |L(c)|=\prod_{j=1}^t |E_{i_j,s_j}|,
 \qquad
 |L(\cD_n)|=\sum_{c\in\cD_n}\prod_{j=1}^t |E_{i_j,s_j}|.
\]
This product formula is the quantity enumerated later by the weighted
skeleton automaton.

\begin{proof}
Suppose
\[
  z=\sigma_S(x)=\sigma_T(x')
\]
for $x,x'\in L(\cD_n)$ and two swap sets $S,T$.  Since the color projection
commutes with coordinate
permutations,
\[
  \chi(z)=\sigma_S(\chi(x))=\sigma_T(\chi(x')).
\]
The skeleton code $\cD_n$ is correcting, hence
\[
  \chi(x)=\chi(x')=:c.
\]

Let
\[
  S^{\times}=\{j\in S:c_j\ne c_{j+1}\},
  \qquad
  T^{\times}=\{j\in T:c_j\ne c_{j+1}\}
\]
be the transpositions crossing the interfaces between maximal runs.  Swaps inside a
constant run do not change the color word, and therefore
\[
  \sigma_{S^{\times}}(c)
  =\chi(z)
  =\sigma_{T^{\times}}(c).
\]
We claim that $S^{\times}=T^{\times}$.  Otherwise, let
\[
  j=\min(S^{\times}\triangle T^{\times})
\]
and assume without loss of generality that
$j\in S^{\times}\setminus T^{\times}$.  The two swap sets agree strictly
to the left of $j$.  In particular, neither contains \(j-1\): the set
\(S^{\times}\) cannot contain it because it contains \(j\), and agreement to
the left gives the same conclusion for \(T^{\times}\).  Thus the symbol at
coordinate $j$ is $c_{j+1}$ after the \(S^{\times}\)-swaps and is $c_j$ after
the \(T^{\times}\)-swaps.  Because $j$ is an interface between two runs,
\[
  c_j\ne c_{j+1},
\]
a contradiction.  Thus
\[
  S^{\times}=T^{\times}.
\]

Write \(S=S^{\times}\sqcup S^{\circ}\) and
\(T=T^{\times}\sqcup T^{\circ}\), where the superscript \(\circ\) denotes
the swaps internal to the runs.  The swaps in each set are disjoint, so they
commute.  Applying the inverse of the common permutation
\(\sigma_{S^{\times}}=\sigma_{T^{\times}}\) to the common output leaves
\[
 \sigma_{S^{\circ}}(x)=\sigma_{T^{\circ}}(x').
\]
All swaps now act strictly inside maximal runs of $c$.  Restricting this common
word to run \(j\) shows that the two words
$x^{(j)},x'^{(j)}\in E_{i_j,s_j}$ have intersecting radius-one balls.  Since
$E_{i_j,s_j}$ is correcting,
\[
  x^{(j)}=x'^{(j)}
\]
for every $j$.  Therefore $x=x'$.
\end{proof}

For example, when \(m\ge3\), the skeleton word \(001120\) has run
decomposition \(0^2 1^2 2^1 0^1\).  A lift chooses independently from
\(E_{0,2},E_{1,2},E_{2,1},E_{0,1}\); the theorem shows that these local
choices remain correcting after the runs are concatenated.

\medskip\noindent\textbf{Inner codes for individual runs.}\par\nopagebreak[4]

Let $\Pi_s\subseteq S_s$ be a correcting permutation code for the same
radius-one adjacent-transposition channel, and write
\[
  M_s=|\Pi_s|.
\]
For an alphabet $A$ of size $Q\ge s$, fix a total order on $A$ and define
$E_s(A)$ as follows.  Choose an $s$-subset $U\subseteq A$, list its elements
in that order, and permute this ordered list according to a codeword of
$\Pi_s$.  This is a correcting code: two words arising from different subsets
have different multisets, which the channel preserves, while words arising
from the same subset are a relabeling of $\Pi_s$.  Moreover,
\[
  |E_s(A)|=M_s\binom Qs.
\]
For fixed $s$ and $Q\to\infty$,
\[
  |E_s(A)|
  =\left(\frac{M_s}{s!}+O(Q^{-1})\right)Q^s.
\]

For \(1\le s\le10\), we use finite permutation codes with the sizes shown in
the first two rows of the table below.  Their correction property was verified
exactly: for every distinct \(\pi,\sigma\in\Pi_s\), all allowed matchings were
enumerated and \(\B(\pi)\cap\B(\sigma)=\emptyset\) was checked.  These checks
establish the listed code sizes; no optimality is claimed.  The code data and
their verification status are recorded in the accompanying reproducibility
package~\cite{ReproPackage2026}.

To obtain lengths \(11\) through \(14\) from the length-ten code, we insert one
new symbol.  The next lemma states the resulting size; its proof defines the
construction and explains the spacing by three.
\begin{lemma}[Separated insertion]
\label{lem:runwise-insert}
If \(\Pi_n\subseteq S_n\) is a correcting permutation code of size \(M_n\),
then there exists a correcting code in \(S_{n+1}\) of size
\[
 \left\lceil\frac{n+1}{3}\right\rceil M_n.
\]
\end{lemma}

\begin{proof}
Introduce a distinguished new symbol $\star$.  For a permutation
$\pi\in S_n$ and $j\in\{0,1,\ldots,n\}$, slot $j$ means the position after
the first $j$ symbols of $\pi$; slot $0$ is before the first symbol and slot
$n$ is after the last.  Let $\operatorname{Ins}_j(\pi)$ be the word obtained
by inserting $\star$ in slot $j$.

Deleting $\star$ from any word in
$\B(\operatorname{Ins}_j(\pi))$ produces a word in $\B(\pi)$.  Indeed, a
transposition involving $\star$ disappears after deletion, while every other
transposition still swaps the same pair of original symbols.  The projected
swaps remain pairwise disjoint.

Choose a set $J\subseteq\{0,1,\ldots,n\}$ of insertion slots with pairwise
distance at least $3$.  For example, take every third slot beginning with
zero.  This gives
\[
  |J|=\left\lceil\frac{n+1}{3}\right\rceil.
\]
Consider
\[
  \{\operatorname{Ins}_j(\pi):\pi\in\Pi_n,\ j\in J\}.
\]
If two balls in this family intersect, deleting $\star$ gives intersecting
balls of the underlying permutations.  Since $\Pi_n$ is correcting, the two
underlying permutations are equal.  The symbol $\star$ can move by at most one
coordinate under a radius-one channel action.  Thus, when two insertion slots
are at distance at least $3$, the possible output positions of $\star$ are
disjoint.  The two insertion slots must therefore also be equal, so the two
centers are identical.
\end{proof}

Applying Lemma~\ref{lem:runwise-insert} to $M_{10}=21766$ yields
\[
  M_{11}=87064,
  \qquad
  M_{12}=348256,
  \qquad
  M_{13}=1741280,
  \qquad
  M_{14}=8706400.
\]
Thus the complete short-run table is
\[
\begin{gathered}
\begin{array}{@{}c|rrrrr@{}}
 s&1&2&3&4&5\\ \hline
 M_s&1&1&2&4&14
\end{array}\\[3pt]
\begin{array}{@{}c|rrrrr@{}}
 s&6&7&8&9&10\\ \hline
 M_s&48&189&859&4116&21766
\end{array}\\[3pt]
\begin{array}{@{}c|rrrr@{}}
 s&11&12&13&14\\ \hline
 M_s&87064&348256&1741280&8706400
\end{array}
\end{gathered}
\]
For the asymptotic application below, take \(q\) large enough that every color
class has at least fourteen symbols.  For short runs we then set
\(E_{i,s}=E_s(A_i)\) for \(1\le s\le14\).
This supplies the specialized inner densities used below.  Longer runs require
a construction that remains available even when the run length exceeds
\(|A_i|\).  For every $s\ge15$, put \(Q=|A_i|\), split
\(A_i=B_{i,0}\sqcup B_{i,1}\) as evenly as possible, and let \(E_{i,s}\)
contain all words whose symbols alternate between the two classes: positions
\(1,3,5,\ldots\) use \(B_{i,0}\), and positions \(2,4,6,\ldots\) use
\(B_{i,1}\).  This code is correcting by
Theorem~\ref{thm:runwise-lifting}, applied inside \(A_i\) to the singleton
binary skeleton \(0101\cdots\) and the length-one inner codes
\(B_{i,0}\) and \(B_{i,1}\).  Its cardinality is
\[
 |B_{i,0}|^{\lceil s/2\rceil}|B_{i,1}|^{\lfloor s/2\rfloor}
 =\bigl(2^{-s}+O_s(Q^{-1})\bigr)Q^s
 \qquad\text{for each fixed }s.
\]

Thus the specialized permutation codes in the table handle run lengths up to
\(14\), and the uniform alternating construction supplies an inner code for
every run length from \(15\) onward.

\medskip\noindent\textbf{Weighted enumeration of skeleton runs.}\par\nopagebreak[4]

We take the skeleton family from the explicit prefix-free ternary block set
\(P^{(3)}_{14}\) of Theorem~\ref{thm:rates}.  If
\(p_\ell=|P^{(3)}_{14}\cap\{0,1,2\}^\ell|\), its nonzero length profile is
\[
 (p_3,p_4,p_6,p_8,p_{10},p_{11},p_{12},p_{13},p_{14})
 =(3,8,18,15,22,3,57,27,112).
\]
The complete list of its \(265\) blocks appears in
Table~\ref{tab:p3-265-blocks} of Appendix~\ref{app:blocks}.  By
Theorem~\ref{thm:rates},
\[
 \cD_n=C_n(P^{(3)}_{14})
\]
is a correcting ternary skeleton code for every \(n\).

The size formula after Theorem~\ref{thm:runwise-lifting} shows what remains to
be counted.  For a skeleton
\(c=i_1^{s_1}\cdots i_t^{s_t}\), its lift contributes the product
\(\prod_j|E_{i_j,s_j}|\), so two features matter: the colors of its runs and
their lengths.  Suppose first that \(|A_i|/q\to\alpha_i\).  For every fixed
short-run length \(s\le14\), define
\[
 \delta_s=\frac{M_s}{s!}.
\]
Then
\[
 \frac{|E_{i,s}|}{q^s}
 =\bigl(\delta_s+o(1)\bigr)\alpha_i^s,
\]
whereas a fixed long run has normalized contribution
\((2^{-s}+o(1))\alpha_i^s\).  Thus each occurrence of color \(i\) contributes
a factor \(\alpha_i\), and closing a short run of length \(s\) contributes the
additional density \(\delta_s\).  A long run instead contributes one factor
\(1/2\) for each of its positions.

We use the rational proportions
\[
 (\alpha_0,\alpha_1,\alpha_2)
 =\frac1{10^9}(97{,}779{,}134,795{,}388{,}998,106{,}831{,}868).
\]
They sum to one.  They were selected by numerical search and then fixed as the
displayed rationals; the proof below uses only exact rational arithmetic and
does not require an optimizer.

To sum the run products over every skeleton in \(\cD_n\), start with the
prefix automaton of \(P^{(3)}_{14}\).  Its nodes are the empty prefix
\(\epsilon\) and all nonempty proper prefixes of its blocks.  From a node
\(\eta\), reading a color \(a\in\{0,1,2\}\) moves to \(\eta a\) if this is a
proper block prefix, returns to \(\epsilon\) if \(\eta a\) completes a block,
and is forbidden otherwise.  Since the block set is prefix-free, a path from
\(\epsilon\) back to \(\epsilon\) labeled by \(n\) colors is exactly a word
of \(C_n(P^{(3)}_{14})\).

We augment each prefix node by the information needed to evaluate the current
run.  After at least one color has been read, a state is a triple
\[
 (\eta,c,s).
\]
Here \(\eta\) is the current prefix-automaton node, \(c\) is the last color
read, and \(s\in\{1,\ldots,15\}\) is the current run length, capped at
\(15\); the value \(s=15\) means that the run has length at least \(15\).  A
separate initial state represents the empty word.  Reading a color \(a\)
simultaneously follows the corresponding prefix-automaton edge and updates
\((c,s)\).

The transition weight records the normalized inner-code contribution.  Every
new color \(a\) contributes \(\alpha_a\).  When a color change closes a short
run of length \(s\le14\), the transition also contributes \(\delta_s\).  If a
run reaches length \(15\), the transition from \(s=14\) to \(s=15\) charges
\(2^{-15}\) for its first fifteen positions; each later same-color transition
charges one further factor \(1/2\).  Hence the complete transition table is
\[
\begin{array}{@{}ll@{}}
\toprule
\text{case}&\text{weight}\\
\midrule
\text{first color }a&\alpha_a\\
 a=c,\ s<14&\alpha_a\\
 a=c,\ s=14&\alpha_a2^{-15}\\
 a=c,\ s=15&\alpha_a/2\\
 a\ne c,\ s\le14&\delta_s\alpha_a\\
 a\ne c,\ s=15&\alpha_a\\
\bottomrule
\end{array}
\]
If a complete path ends with a short run of length \(s\), it receives the
terminal factor \(\delta_s\); a final long run receives terminal factor one.
These terminal factors ensure that every run is counted, including the last
one.

Retain the states reachable from the initial state and from which some state
with prefix component \(\epsilon\) is reachable.  Let \(V\) be this set after
the transient initial state is removed, and define \(W_{uv}\) as the sum of
the weights of all legal transitions from \(u\) to \(v\).  The resulting
nonnegative rational matrix is
\[
 W\in\mathbb Q_{\ge0}^{V\times V},
 \qquad |V|=1013,
\]
and it has \(1527\) nonzero entries.  These \(1013\) states are precisely the
feasible combinations of a proper-prefix node and a run memory that remain
after the stated reachability pruning.  For each fixed \(n\), the total weight of the
length-\(n\) paths that finish at a block boundary, including the terminal
factor, is the termwise large-\(q\) limit of \(q^{-n}|L(\cD_n)|\).  Thus the
spectral radius of \(W\) is the limiting normalized exponential growth factor
we need.

For a square matrix \(N\), write \(\rho(N)\) for its spectral radius.  An
exact-integer computation produces a vector \(h\in\mathbb Z_{>0}^{V}\) and
verifies
\[
 W^{\mathsf T}h
 >\frac{100{,}000{,}000}{182{,}560{,}995}\,h
 \qquad\text{coordinatewise}.
\]
The matrix, the vector, and the deterministic verifier are included in the
accompanying reproducibility package~\cite{ReproPackage2026}.  The
Collatz--Wielandt inequalities now give
\[
 \rho(W)>\frac{100{,}000{,}000}{182{,}560{,}995},
 \qquad
 \rho(W)^{-1}<1.82560995.
\]

It remains to connect this limiting weighted model to an actual code over each
large finite alphabet.  Choose nonempty color classes with \(|A_i|\ge14\) and
\[
 \alpha_{i,q}:=\frac{|A_i|}{q}\longrightarrow\alpha_i.
\]
For \(1\le s\le14\), put
\[
 \delta_{i,s}^{(q)}
 =\frac{|E_s(A_i)|}{|A_i|^s},
\]
and, for the balanced split
\(A_i=B_{i,0}\sqcup B_{i,1}\), put
\[
 \beta_{i,q}
 =\frac{\min\{|B_{i,0}|,|B_{i,1}|\}}{|A_i|}.
\]
Define the finite-\(q\) matrix \(W_q\) from the same state graph by replacing
\(\alpha_i\) with \(\alpha_{i,q}\), \(\delta_s\) with
\(\delta_{i,s}^{(q)}\) when a run of color \(i\) closes, \(2^{-15}\) with
\(\beta_{i,q}^{15}\), and \(1/2\) with \(\beta_{i,q}\).  These quantities
converge to their counterparts in \(W\), so \(W_q\to W\) entrywise.

For each \(c\in C_n(P^{(3)}_{14})\), lift its short runs with the codes
\(E_s(A_i)\) and its long runs with the alternating codes, and denote the union
of these lifts by \(L_{q,n}\).  Theorem~\ref{thm:runwise-lifting} makes
\(L_{q,n}\) correcting.  For a long run of color \(i\), the alternating code
has at least \((\beta_{i,q}|A_i|)^s\) words, so complete-path weights in
\(W_q\), including the short terminal factor, lower-bound
\(q^{-n}|L_{q,n}|\).  Every retained state is reachable and can reach a block
boundary.  Standard Perron--Frobenius path enumeration therefore realizes the
spectral growth of \(W_q\), up to bounded initial and terminal factors, on
complete skeleton paths.  It follows that
\[
 \limsup_{n\to\infty}\frac{|L_{q,n}|^{1/n}}q\ge\rho(W_q).
\]
Equivalently, this correcting family has rate at least
\[
 1+\log_q\rho(W_q).
\]
Thus \(\rho(W_q)^{-1}\) is the multiplicative loss factor of the finite-\(q\)
weighted construction, and it converges to \(\rho(W)^{-1}\).

\begin{theorem}[Improved large-alphabet loss factor]
\label{thm:runwise-main}
For the radius-one limited permutation channel,
\[
 C_0^{(q)}
 \ge
 1-\log_q\!\left(\frac{182{,}560{,}995}{100{,}000{,}000}\right)
 +o\!\left(\frac1{\ln q}\right).
\]
The underlying spectral loss factor is strictly below \(1.82560995\).  Thus,
relative to the bound \(1-\log_q2\) of Chee et al.~\cite{CKLNVZ16}, the
construction improves the leading multiplicative loss factor from \(2\) to a
value below \(1.82560995\).
\end{theorem}
\begin{proof}
By Theorem~\ref{thm:runwise-lifting}, every code \(L_{q,n}\) constructed above
is correcting.  The weighted path estimate gives an achieved rate at least
\(1+\log_q\rho(W_q)\).  Since \(W_q\to W\) entrywise,
\(\rho(W_q)=\rho(W)+o(1)\); hence
\[
 1+\log_q\rho(W_q)
 =1+\log_q\rho(W)+o(1/\ln q).
\]
The exact coordinatewise inequality for \(h\), followed by the
Collatz--Wielandt bound, gives
\[
 \rho(W)>\frac{100{,}000{,}000}{182{,}560{,}995},
\]
which yields the stated result.
\end{proof}

The covering result upper-bounds the capacity and therefore lower-bounds the
normalized loss \(\Delta_q\).  The run-wise construction lower-bounds the
capacity and therefore upper-bounds \(\Delta_q\).  Together they give the
following quantitative interval.

\begin{corollary}[Large-alphabet correction sandwich]
\label{cor:correction-sandwich}
With \(\Delta_q=(1-C_0^{(q)})\ln q\) as defined at the beginning of the
section, we have
\[
 \ln\varphi
 \le \liminf_{q\to\infty}\Delta_q
 \le \limsup_{q\to\infty}\Delta_q
 \le \ln\!\left(\frac{182{,}560{,}995}{100{,}000{,}000}\right).
\]
Equivalently,
\[
 \varphi
 \le \liminf_{q\to\infty}q^{1-C_0^{(q)}}
 \le \limsup_{q\to\infty}q^{1-C_0^{(q)}}
 \le 1.82560995.
\]
\end{corollary}

\begin{proof}
The left inequality is Theorem~\ref{thm:covering-phi}.  Theorem~\ref{thm:runwise-main}
gives the right inequality.  The equivalent form follows from
\(q^{1-C_0^{(q)}}=e^{\Delta_q}\).
\end{proof}


\section{Error detection}
\label{sec:detection}

Detection is useful when a receiver can recognize that an error occurred and
request retransmission, without identifying the transmitted codeword.  Thus
two error balls may intersect away from the code, but no nontrivial channel
output may itself be another codeword.  For distinct \(u,v\), correction
forbids \(\B(u)\cap\B(v)\ne\emptyset\), whereas detection forbids
\(v\in\B(u)\).  The latter condition is naturally directed.  For the present
channel the full-word relation is symmetric because every allowed swap is an
involution, but its direction remains important when we examine truncated
prefixes.

We first determine the optimal first-order large-alphabet detection loss.  A
stable labeling of repeated symbols converts each fixed-composition class into
a subset of a permutation space; pulling back permutation codes then gives an
all-alphabet finite-length construction.  We next give explicit weak-zigzag
codes that improve the resulting fixed-alphabet bounds for \(q=3,4\).  For
block-concatenation codes, we define a directed extinction depth that measures
how far a finite certification must inspect each ordered pair of first blocks.
This yields an all-length detecting criterion, certifies a block set left
unresolved by the earlier local test, and gives an exact linear depth law for
the two-block corridor family.

\subsection{The optimal large-alphabet detection loss}
\label{sec:det-rates}

Let \(D_q(n)\) be the largest size of a detecting code in \(\Sigma_q^n\),
and put
\[
 R_{\mathrm{det},q}
 =
 \limsup_{n\to\infty}\frac1n\log_qD_q(n).
\]
As in the correction section, it is useful to normalize the gap from rate one:
\[
 \Delta_q^{\mathrm{det}}
 =(1-R_{\mathrm{det},q})\ln q,
 \qquad
 q^{1-R_{\mathrm{det},q}}=e^{\Delta_q^{\mathrm{det}}}.
\]

The full adjacent-pairing argument appears in
\cite[Sec.~7.1]{PriorWork2026}.  Its idea is short.  Partition the coordinates
into \((1,2),(3,4),\ldots\), and record only the unordered multiset in each
pair.  If two words have the same record, one is obtained from the other by
swapping some of these disjoint pairs, so a detecting code contains at most
one word with each record.  Since a pair has \(\binom{q+1}{2}\) possible
unordered contents, this gives the finite-length bound
\[
 D_q(n)
 \le
 q^{\,n-2\lfloor n/2\rfloor}
 \binom{q+1}{2}^{\lfloor n/2\rfloor}
\]
and hence
\begin{equation}
 \label{eq:det-pairing-upper}
 R_{\mathrm{det},q}
 \le
 \frac12\log_q\binom{q+1}{2}.
\end{equation}
We now construct detecting codes over every alphabet that match the
first-order term of this upper bound.

The construction has two steps.  First consider a permutation word, meaning a
word containing every label in \([n]\) exactly once.  For \(\pi\in S_n\),
write
\(\operatorname{pos}_{\pi}(a)=\pi^{-1}(a)\) for the position of label
\(a\).

After solving this permutation problem, we return to arbitrary words.  Inside
each class of words having the same symbol multiplicities, equal occurrences
are labeled from left to right.  Every effective adjacent swap then becomes an
adjacent swap of distinct labels, allowing the permutation code to be pulled
back to the original alphabet.

\begin{proposition}[The detecting permutation subproblem]
\label{prop:perm-detect}
For \(\pi,\sigma\in S_n\),
\[
 \sigma\in\B(\pi)
 \quad\Longleftrightarrow\quad
 \max_{a\in[n]}
 \left|
   \operatorname{pos}_{\pi}(a)-\operatorname{pos}_{\sigma}(a)
 \right|
 \le1.
\]
Consequently, the maximum size of a detecting permutation code is
\(n!/2^{\lfloor n/2\rfloor}\).
\end{proposition}

\begin{proof}
A matching of adjacent position swaps moves every label by at most one
position, proving the forward implication.  Conversely, define the
old-to-new position map \(\alpha\in S_n\) by
\[
 \alpha\bigl(\operatorname{pos}_{\pi}(a)\bigr)
 =\operatorname{pos}_{\sigma}(a)
 \qquad(a\in[n]).
\]
The displayed distance condition says \(\abs{\alpha(i)-i}\le1\) for every
position \(i\).  Scan positions from left to right.  At the first unprocessed
position \(i\), either \(\alpha(i)=i\), or \(\alpha(i)=i+1\); in the latter
case bijectivity and the displacement bound force
\(\alpha(i+1)=i\).  Removing the fixed point or this adjacent transposition
and continuing proves that \(\alpha\) is a product of pairwise disjoint
adjacent transpositions.  Hence \(\sigma\in\B(\pi)\).

The map \(\pi\mapsto\pi^{-1}\) sends a permutation word to its vector of label
positions.  By the equivalence just proved, a detecting permutation code is
therefore exactly a permutation array whose distinct elements have Chebyshev
distance at least two.

Bereg, Haghpanah, Malouf, and Sudborough proved that the maximum size of a
permutation array of minimum Chebyshev distance two is
\(n!/2^{\lfloor n/2\rfloor}\)~\cite[Thm.~7]{BHMS24}.  Inverting an optimal
array gives a detecting code of this size, and the same equivalence proves
optimality.
\end{proof}

We next transfer such a permutation code to words with repeated symbols.  A
composition vector \(\boldsymbol\nu=(\nu_0,\ldots,\nu_{q-1})\) records how
many times each alphabet symbol occurs; its type class has
\(n!/(\nu_0!\cdots\nu_{q-1}!)\) words.

\begin{theorem}[All-alphabet detecting sandwich]
\label{thm:detqary}
For every \(q\ge2\) and \(n\ge1\),
\begin{equation}
\label{eq:det-finite-lower}
 D_q(n)
 \ge
 \sum_{\substack{\nu_0+\cdots+\nu_{q-1}=n\\ \nu_i\ge0}}
 \left\lceil
 \frac{n!}
 {\nu_0!\cdots\nu_{q-1}!\,2^{\lfloor n/2\rfloor}}
 \right\rceil
 \ge
 \frac{q^n}{2^{\lfloor n/2\rfloor}}.
\end{equation}
Hence
\begin{equation}
  1-\frac12\log_q2
  \le
  R_{\mathrm{det},q}
  \le
  1-\frac12\log_q2
  +\frac12\log_q\!\left(1+\frac1q\right).
\end{equation}
In particular, as \(q\to\infty\),
\begin{equation}
\label{eq:det-optimal-loss}
 R_{\mathrm{det},q}
 =
 1-\frac{\ln2}{2\ln q}
 +O\!\left(\frac1{q\ln q}\right)
 =
 1-\log_q\sqrt2
 +O\!\left(\frac1{q\ln q}\right).
\end{equation}
Equivalently,
\[
 \Delta_q^{\mathrm{det}}=\ln\sqrt2+O(q^{-1}),
 \qquad
 q^{1-R_{\mathrm{det},q}}=\sqrt2\bigl(1+O(q^{-1})\bigr).
\]
Thus \(\sqrt2\) is the optimal first-order large-alphabet loss factor for
detection.
\end{theorem}

\begin{proof}
Fix a composition vector
\(\boldsymbol\nu=(\nu_0,\ldots,\nu_{q-1})\), and let
\[
 \cT_{\boldsymbol\nu}
 =\{x\in\Sigma_q^n:\text{symbol }a\text{ occurs }\nu_a
   \text{ times for every }a\}
\]
be its type class.  Distinguish equal occurrences stably from left to right:
the \(k\)-th occurrence of symbol \(a\) receives the label \((a,k)\).  For
example, \(1011\) is labeled by the sequence
\(((1,1),(0,1),(1,2),(1,3))\).  Choose once and for
all a bijection from the \(n\) labels of this type to \([n]\).  Let
\(\ell_{\boldsymbol\nu}(x)\in S_n\) be the resulting permutation word and put
\[
 \Omega_{\boldsymbol\nu}
 =\ell_{\boldsymbol\nu}(\cT_{\boldsymbol\nu}).
\]
Thus the stable lift is a bijection
\[
 \ell_{\boldsymbol\nu}:\cT_{\boldsymbol\nu}
 \longrightarrow
 \Omega_{\boldsymbol\nu}\subseteq S_n.
\]
If \(y=\sigma_S(x)\), define
\[
 S_{\mathrm{eff}}=\{i\in S:x_i\ne x_{i+1}\}.
\]
The deleted swaps exchange equal symbols and hence do not change the word.
Every remaining swap exchanges unequal symbols, so it preserves the
left-to-right order of all occurrences carrying the same symbol.  Therefore
\[
 \ell_{\boldsymbol\nu}(y)
 =
 \sigma_{S_{\mathrm{eff}}}
 \bigl(\ell_{\boldsymbol\nu}(x)\bigr).
\]
If \(x\ne y\), then \(S_{\mathrm{eff}}\ne\varnothing\), so the two stable
lifts are distinct and are related by a nonempty matching of adjacent
transpositions.  Consequently, any detecting subset of
\(\Omega_{\boldsymbol\nu}\) pulls back to a detecting subset of the type
class.

Let \(A\subseteq S_n\) be a detecting permutation code of size
\(n!/2^{\lfloor n/2\rfloor}\), supplied by
Proposition~\ref{prop:perm-detect}.  For \(g\in S_n\), let
\(gA=\{g\circ\pi:\pi\in A\}\); this changes the names of the labels but not
their positions.  Value relabeling commutes with coordinate swaps, hence every
\(gA\) remains detecting.  Count the pairs
\((g,\rho)\) with \(\rho\in\Omega_{\boldsymbol\nu}\cap gA\).  For each
\(\rho\in\Omega_{\boldsymbol\nu}\) and \(\pi\in A\), exactly one relabeling,
namely \(g=\rho\circ\pi^{-1}\), satisfies \(g\circ\pi=\rho\).  Dividing this
count by the \(n!\) choices of \(g\) gives
\[
 \frac1{n!}\sum_{g\in S_n}
 \abs{\Omega_{\boldsymbol\nu}\cap gA}
 =
 \frac{
   \abs{\Omega_{\boldsymbol\nu}}\abs A
 }{n!}
 =
 \frac{
   \abs{\cT_{\boldsymbol\nu}}
 }{2^{\lfloor n/2\rfloor}}.
\]
Choose a translate attaining at least the average and pull its intersection
back through \(\ell_{\boldsymbol\nu}\).  Since its size is an integer, it is
at least the ceiling in \eqref{eq:det-finite-lower}.

The channel preserves type, so these subsets may be united over all
\(\boldsymbol\nu\).  Summing the type-class sizes gives \(q^n\), proving
the finite-length lower bound and then
\(R_{\mathrm{det},q}\ge1-\frac12\log_q2\).
Finally, rewrite \eqref{eq:det-pairing-upper} as
\[
 \frac12\log_q\binom{q+1}{2}
 =
 1-\frac12\log_q2
 +\frac12\log_q\!\left(1+\frac1q\right).
\]
The last term is \(O(1/(q\ln q))\), which proves the sandwich and
\eqref{eq:det-optimal-loss}.
\end{proof}

Every correcting code is also detecting, so
Theorem~\ref{thm:runwise-main} already supplies detecting codes with asymptotic
loss factor below \(1.82560995\).  The detection-specific type-class lifting
improves that factor to \(\sqrt2\), to first order, and the pairing upper bound
proves that the coefficient \(\ln2/2\) is optimal.

\subsection{Explicit fixed-alphabet detection}

The type-class argument is uniform and first-order sharp, but it selects a
code by averaging.  For small alphabets a simple local rule gives stronger
explicit bounds: alternate the direction of a weak inequality, so that every
effective adjacent swap reverses one required inequality.

\begin{theorem}[Weak-zigzag construction]
\label{thm:zigzag}
Order \(\Sigma_q=\{0,1,\ldots,q-1\}\), and define
\[
 Z_{q,n}=\{x\in\Sigma_q^n:x_1\le x_2\ge x_3\le x_4\ge\cdots\}.
\]
Then \(Z_{q,n}\) is detecting for every \(n\), and
\begin{equation}
\label{eq:zigzag-rate}
 R_{\mathrm{det},q}\ge
 \log_q\!\left(\frac{1}{2\sin\!\bigl(\pi/(4q+2)\bigr)}\right).
\end{equation}
In particular,
\[
 R_{\mathrm{det},3}\ge0.736917768,
 \qquad
 R_{\mathrm{det},4}\ge0.762880412.
\]
\end{theorem}

\begin{proof}
Let \(x\in Z_{q,n}\) and \(y=\sigma_S(x)\ne x\).  Some selected site
\(i\in S\) is effective, so \(x_i\ne x_{i+1}\).  No other selected site
touches positions \(i,i+1\).  If the zigzag condition there is
\(x_i\le x_{i+1}\), it is therefore strict, whereas the swapped output has
\(y_i>y_{i+1}\).  The case of a required \(\ge\) inequality is symmetric.
Thus \(y\notin Z_{q,n}\), which proves detection.  Weak inequalities are
essential here because swapping two equal adjacent symbols does not produce an
error.

For the rate, let \(U\) be the \(q\times q\) matrix with, for
\(a,b\in\Sigma_q\), \(U_{ab}=1\) when \(a\le b\) and \(U_{ab}=0\)
otherwise, and let \(\boldsymbol 1\) be the all-one column vector.  Thus
\(U\) records one upward weak step and \(U^{\mathsf T}\) records one downward
weak step.  Multiplying these transition matrices and summing over the first
and last symbols gives, for even lengths \(2m\) with \(m\ge1\) and odd
lengths \(2m+1\) with \(m\ge0\), respectively,
\[
 |Z_{q,2m}|=\boldsymbol 1^{\mathsf T}(UU^{\mathsf T})^{m-1}U\boldsymbol 1,
 \qquad
 |Z_{q,2m+1}|=\boldsymbol 1^{\mathsf T}(UU^{\mathsf T})^m\boldsymbol 1.
\]
The inverse of \(UU^{\mathsf T}\) is tridiagonal, with off-diagonal entries
\(-1\), diagonal entries \(2\), and one endpoint entry \(1\).  Its eigenvalues
are
\[
 4\sin^2\!\left(\frac{(2j-1)\pi}{4q+2}\right),
 \qquad 1\le j\le q.
\]
The matrix \(UU^{\mathsf T}\) is positive.  Perron--Frobenius shows that its
spectral radius controls the exponential growth per two symbols.  Since
\(\rho(UU^{\mathsf T})=\|U\|_2^2\), where \(\|U\|_2\) is the largest singular
value of \(U\), the growth per symbol is
\[
 \lim_{n\to\infty}|Z_{q,n}|^{1/n}
 =\|U\|_2
 =\frac{1}{2\sin\!\bigl(\pi/(4q+2)\bigr)},
\]
which proves~\eqref{eq:zigzag-rate}.
\end{proof}

Table~\ref{tab:det-rates} records the stronger of our two lower bounds for each
of \(q=3,4,5\), together with the pairing upper bound.  Weak zigzags are better
for \(q=3,4\), whereas type-class lifting is better for \(q=5\).

\begin{table}[!t]
\caption{Detecting-rate bounds, outward-rounded to six decimal places.}
\label{tab:det-rates}
\centering
\small
\setlength{\tabcolsep}{5pt}
\begin{tabular}{@{}cclc@{}}
\toprule
\(q\) & lower bound & construction & pairing upper bound\\
\midrule
3 & 0.736917 & weak zigzag & 0.815465\\
4 & 0.762880 & weak zigzag & 0.830483\\
5 & 0.784661 & type-class lifting & 0.841304\\
\bottomrule
\end{tabular}
\end{table}

\subsection{Directed extinction}

We now return from unrestricted fixed-length codes to a fixed block set
\(P\).  The task is offline certification: from the finite description of
\(P\), we verify once that no effective channel error maps one word of
\(C_n(P)\) to another, simultaneously for every \(n\).  After common
initial blocks are cancelled, the verifier follows possible prefixes beginning
with two distinct first blocks.  Detection is directed: one prefix is treated
as the transmitted word, and the other as the candidate codeword that the
channel output might equal.

The depth introduced below is the number of transmitted symbols that this
all-length certificate may have to inspect before such a directed ambiguity
disappears.  It is therefore a complexity parameter of code certification;
decoder complexity is a separate question.

\begin{definition}[Directed live sets and extinction depth]
\label{def:dlive}
For \(n\ge1\) and \(u,v\in\Sigma_q^n\), the ordered pair \((u,v)\) is
\emph{directed-live} if
\[
 \prefx_{n-1}(v)\in\Tball(u),
\]
and \emph{directed-extinct} otherwise.  Equivalently, some allowed channel
output from \(u\) agrees with the candidate word \(v\) in the first \(n-1\)
positions.  The final position is not yet tested because a future adjacent
swap across the right edge of the current window may still change it.

For distinct \(a,b\in P\), put
\[
 \Live_n^{\to}(a,b)=
 \{(u,v):u\in\Pref_n(aP^*),\ v\in\Pref_n(bP^*),
          \ \prefx_{n-1}(v)\in\Tball(u)\},
\]
where these prefix sets contain all legal length-\(n\) continuations beginning
with \(a\) and \(b\), even when the window ends inside a later block.  Define
\[
 \dko(a,b)=\min\{n\ge1:\Live_n^{\to}(a,b)=\emptyset\},
\]
 with value \(\infty\) if the set never empties.  The detecting criterion below
 needs the shorter-first orientation for unequal block lengths and both
 orientations for equal lengths.  When all these relevant depths are finite,
 write
\[
 K_{\mathrm{det},0}(P)=
 \max_{\substack{a,b\in P,\ a\ne b\\|a|\le|b|}}\dko(a,b).
\]
\end{definition}

Liveness is persistent under taking shorter prefixes.  Indeed, a witness at
depth \(m\) gives, by Lemma~\ref{lem:window-import}, a witness for the first
\(n\) symbols for every \(n\le m\).  Consequently, once
\(\Live_n^{\to}(a,b)\) is empty, every later level is empty as well, and
\(\dko(a,b)\) is the first window length that certifies this ordered
first-block pair.  For one fixed equal-length pair, directed extinction is
precisely the same-length detecting test of
\cite[Def.~13]{PriorWork2026}.

\begin{theorem}[Depth-\(K\) detecting criterion]
\label{thm:detdepthK}
Let \(P\subseteq\Sigma_q^*\) be a finite nonempty set of nonempty blocks, and let
\(K\ge1\).  Assume that, for every ordered pair of distinct blocks
\((a,b)\) with \(|a|\le|b|\),
\[
 \tag{DW}
 \Live_K^{\to}(a,b)=\emptyset.
\]
Then \(C_n(P)\) is detecting for every \(n\).
\end{theorem}

For equal-length blocks, the hypothesis includes both orders \((a,b)\) and
\((b,a)\).  For unequal lengths, only the shorter-first order is needed: the
full-word relation \(v\in\B(u)\) is symmetric because every allowed channel
permutation is an involution, even though the truncated live relation is
directed.

\begin{proof}
Suppose \(c'\in\B(c)\) for distinct codewords.  Fix block factorizations and
cancel all common initial blocks using the directed reduction in
Lemma~\ref{lem:cancel-import}.  We obtain equal-length words
\[
 x\in aP^*,\qquad y\in bP^*,\qquad y\in\B(x),
\]
where \(a\ne b\) and their common length \(n'\) is positive.  If
\(|a|>|b|\), exchange \(x\) and \(y\); this preserves the channel relation
because the swaps are involutions.  We may therefore assume
\(|a|\le|b|\).

If \(n'\ge K\), set \(X=x\) and \(Y=y\).  If \(n'<K\), choose any
\(p\in P\) and append the same power \(p^m\) until the common length is at
least \(K\); set \(X=xp^m\) and \(Y=yp^m\).  Then
Lemma~\ref{lem:extend-collision} gives \(Y\in\B(X)\).  In either case,
Lemma~\ref{lem:window-import} at window length \(K\) yields
\[
 \prefx_{K-1}(Y)\in\Tball(\prefx_K(X)).
\]
The two length-\(K\) prefixes belong to \(\Pref_K(aP^*)\) and
\(\Pref_K(bP^*)\), respectively, so they form an element of
\(\Live_K^{\to}(a,b)\), contradicting \textup{(DW)}.
\end{proof}

The earlier two-stage detecting criterion
\cite[Def.~14 and Thm.~17]{PriorWork2026} can be stated in the present
language without importing any additional notation.  It tests an ordered
equal-length pair at depth \(|a|=|b|\).  For \(|a|<|b|\), it first tests the
length-\(|a|\) prefixes and, if they are still live, tests every legal
continuation at depth \(|b|\).  Thus the old test requires each pair to become
extinct no later than \(\max\{|a|,|b|\}\).  Theorem~\ref{thm:detdepthK}
removes this pairwise cutoff: a pair may survive across further block
boundaries, provided that it becomes extinct by one finite common horizon
\(K\).

\medskip\noindent\textbf{Beyond the two-stage cutoff.}
Let \(P_{\mathrm{det}}\) denote the explicit 50-block detecting set listed in
\cite[Sec.~7.2]{PriorWork2026}.  For unequal block lengths, retain only the
shorter-first orientation; for equal lengths, retain both orientations.  This
gives \(1482\) canonical ordered first-block pairs.  An exact exhaustive
traversal finds that all of them become extinct by depth
\(11=\lmax\)~\cite{ReproPackage2026}.  Ten ordered pairs survive beyond the
cutoff used by the old test---two equal-length tests and eight unequal-length
tests---which is why that criterion leaves the set unresolved.  The new depth
criterion follows those pairs farther, certifies the entire all-length family,
and recovers its detecting rate \(0.756707\).  This gives a concrete strict
extension of the previous local certificate.

\subsection{Linear depth on the corridor}
\label{sec:detlin}

Recall the binary family from Theorem~\ref{thm:corridor}:
\[
\begin{aligned}
 a_k&=00(10)^{k-3}1111, & |a_k|&=2k,\\
 b_k&=(01)^{k-1}1110,   & |b_k|&=2k+2,
\end{aligned}
\qquad P_k=\{a_k,b_k\}.
\]
The long alternating parts of \(a_k\) and \(b_k\) occur in opposite phases.
They form a corridor in which local swaps can keep the transmitted prefix
consistent with the target prefix across more than one block boundary.  For
example,
\[
 a_4=00101111,\qquad b_4=0101011110.
\]

The next theorem determines both directed depths for every \(k\ge3\).  This is
stronger than a finite-range computation.  Since \(|a_k|<|b_k|\), the
shorter-first orientation is the one used by the all-length criterion.  On
the range \(3\le k\le60\), where Theorem~\ref{thm:corridor} gives correction
depth \(2k^2+4k-1\), its exact detection-certification depth is only \(4k\).  The difference
comes from the obstruction being one-sided: correction compares two
independently permuted channel outputs, whereas detection compares the output
of one transmitted word with one fixed target codeword.

\begin{theorem}[Family detection law]
\label{thm:detlaw}
For every \(k\ge3\),
\[
 \dko(a_k,b_k)=4k,
 \qquad
 \dko(b_k,a_k)=4k+1.
\]
Consequently, \(K_{\mathrm{det},0}(P_k)=4k\), and
Theorem~\ref{thm:detdepthK} shows that \(C_n(P_k)\) is detecting for every
\(n\).
\end{theorem}

\begin{proof}
Write \(a=a_k\) and \(b=b_k\).  For a block \(x\in\{a,b\}\), the notation
\(x^\infty=xxx\cdots\) denotes its infinite periodic repetition; below we use
only finite prefixes of this word.  After the first block on each side, only
the choice \(x\) of the next transmitted block and the choice \(y\) of the
next target block can affect the windows under consideration.

A directed-live pair at depth \(n\) requires a swap set on the transmitted
prefix whose output agrees with the fixed target in positions
\(1,\ldots,n-1\).  Ineffective swaps of equal symbols may be deleted.  The
remaining decisions can be scanned from left to right.  If site \(i-1\) was
selected, then output position \(i\) is already forced to use input position
\(i-1\).  Otherwise one may either leave position \(i\) unchanged, producing
the current input symbol, or select site \(i\), producing the next input
symbol and consuming both positions.  When exactly one choice matches the
target symbol, that choice is forced; when neither choice matches, the pair
becomes extinct at depth \(i+1\).

For \(k=3\), this forced scan gives the first empty depth in each of the eight
continuation cases:
\[
\begin{array}{@{}c|cccc@{}}
\text{first blocks}\backslash(x,y)
 &(a,a)&(a,b)&(b,a)&(b,b)\\ \hline
(a,b)&8&8&11&12\\
(b,a)&9&12&9&13
\end{array}
\]
The row maxima are \(12=4k\) and \(13=4k+1\), respectively.  We henceforth
assume \(k\ge4\).  The same scan has the following six outcomes; the remainder
of the proof verifies each row and supplies a surviving swap set one position
before extinction.
\[
\begin{array}{@{}ccl@{}}
\text{first blocks}&\text{continuation condition}&\text{first empty depth}\\ \hline
(a,b)&x=a&2k+2\\
(a,b)&x=b,\ y=a&4k-1\\
(a,b)&x=b,\ y=b&4k\\
(b,a)&y=a&2k+3\\
(b,a)&x=a,\ y=b&4k\\
(b,a)&x=b,\ y=b&4k+1
\end{array}
\]

Below, saying that a displayed swap set gives survival through depth \(d\)
means that it makes the transmitted output agree with the target in the first
\(d-1\) positions, and hence proves that the corresponding directed live set
at depth \(d\) is nonempty.

\medskip\noindent\emph{The order \((a_k,b_k)\).}
At depth \(4k\), the relevant prefixes of every legal pair beginning with
\((a_k,b_k)\) are represented by
\[
 \bigl(\prefx_{4k}(a_kx^\infty),
       \prefx_{4k}(b_ky^\infty)\bigr),
 \qquad x,y\in\{a_k,b_k\}.
\]
Let \(S_{\mathrm{alt}}=\{2,4,\ldots,2k-4\}\).  On the initial alternating
segment, the transmitted and target phases are opposite.  At each even site,
only the next input symbol matches the target, so the scan forces precisely
the sites in \(S_{\mathrm{alt}}\).

If \(x=a\), then at position \(2k+1\) the transmitted symbol is \(0\) and the
target symbol is \(1\).  The preceding input \(1\) is already forced to
remain at position \(2k\), while the following input symbol is again \(0\).
Thus both choices of \(y\) become extinct at depth \(2k+2\), and
\(S_{\mathrm{alt}}\) shows survival through depth \(2k+1\).

Let \(x=b\).  The mismatch at position \(2k+1\) forces swap site \(2k+1\).
If \(y=a\), the next opposite-phase alternating segment forces
\[
 T=\{2k+4,2k+6,\ldots,4k-4\}.
\]
At position \(4k-2\), the current input is \(1\) and the target symbol is
\(0\).  The preceding \(0\) has been consumed by the swap at site \(4k-4\),
and the following input is \(1\).  Hence this case becomes extinct at depth
\(4k-1\), while
\(S_{\mathrm{alt}}\cup\{2k+1\}\cup T\) shows survival through depth
\(4k-2\).  If \(y=b\), the two copies following the first blocks are shifted
by two positions and agree through position \(4k-2\).  At position \(4k-1\),
the current input is \(1\), the target symbol is \(0\), and the only adjacent
input symbols still available are both \(1\).  This case becomes extinct at
depth \(4k\), while
\(S_{\mathrm{alt}}\cup\{2k+1\}\) shows survival through depth \(4k-1\).
Therefore \(\dko(a_k,b_k)=4k\).

\medskip\noindent\emph{The reverse order \((b_k,a_k)\).}
Put
\[
 U_x=b_kx^\infty,
 \qquad V_y=a_ky^\infty,
 \qquad x,y\in\{a_k,b_k\}.
\]
Every legal pair beginning with \((b_k,a_k)\) and having length \(4k+1\) has
the same relevant prefixes as \((U_x,V_y)\) for some \(x,y\).  The transmitted
stream is still inside its second block.  On the target side, if the second
block is \(a\), the window uses one symbol from a third block; every legal block
begins with \(0\), so representing that continuation by \(a^\infty\) does not
change the prefix.

The common initial scan forces
\[
 S_0=\{2,4,\ldots,2k-4\}\cup\{2k+1\}.
\]
If \(y=a\), the swap at site \(2k+1\) produces \(01\) at positions
\(2k+1,2k+2\), whereas the next block of the target begins with \(00\).
Both choices of \(x\) therefore become extinct at depth \(2k+3\), and \(S_0\)
shows survival through depth \(2k+2\).

Now let \(y=b\).  If \(x=a\), the following opposite-phase segment forces
\[
 T'=\{2k+4,2k+6,\ldots,4k-2\}.
\]
The last swap outputs \(0\) at position \(4k-1\), where the target
symbol is \(1\).  Hence this case becomes extinct at depth \(4k\), while
\(S_0\cup T'\) shows survival through depth \(4k-1\).  Finally, if
\(x=b\), the two copies of \(b\) are shifted by two positions and agree
through position \(4k-2\).  At position \(4k-1\), swap site \(4k-1\) is forced;
it then outputs \(0\) at position \(4k\), where the required symbol is \(1\).
Thus this case becomes extinct at depth \(4k+1\), while
\(S_0\cup\{4k-1\}\) shows survival through depth \(4k\).  Therefore
\(\dko(b_k,a_k)=4k+1\), completing the proof.
\end{proof}

The preceding proof used the explicit bit patterns of \(a_k\) and \(b_k\).
The next two local statements isolate the reusable mechanism.  The first
records exactly where each output symbol came from; the second applies that
record to alternating binary intervals.

\begin{lemma}[Displacement normal form]
\label{lem:domino}
Let \(u,v\in\Sigma_q^n\) and \(v=\sigma_S(u)\) for a swap set \(S\).  For each
output coordinate \(i\), define its source offset
\[
 d(i)=\begin{cases}
 +1,&i\in S,\\
 -1,&i-1\in S,\\
 0,&\text{otherwise}.
 \end{cases}
\]
Then \(d(i)\in\{-1,0,+1\}\) and
\[
 v_i=u_{i+d(i)}.
\]
Moreover, the support \(\{i:d(i)\ne0\}\) is a disjoint union of adjacent pairs
\((i,i+1)\) with \(d(i)=+1\) and \(d(i+1)=-1\).  Conversely, every function of
this paired form comes from the swap set \(\{i:d(i)=+1\}\).  Thus a nonzero
displacement is confined to one exchanged adjacent pair; it cannot be passed
from one pair to another.
\end{lemma}
\begin{proof}
If site \(i\) is swapped, output coordinate \(i\) receives input coordinate
\(i+1\), giving offset \(+1\).  If site \(i-1\) is swapped, output coordinate
\(i\) receives input coordinate \(i-1\), giving offset \(-1\).  Otherwise it
receives input coordinate \(i\).  Consecutive sites cannot both belong to
\(S\), so these alternatives are disjoint and every nonzero offset occurs in
the stated \(+1,-1\) pair.
\end{proof}

\begin{lemma}[Alternating intervals and constant runs]
\label{lem:rigidity}
Let \(q=2\), \(v=\sigma_S(u)\), and let \(d\) be the source-offset function of
Lemma~\ref{lem:domino}.  Say that \(u\) alternates on an interval
\(I=\{r,\ldots,t\}\) if \(u_j\ne u_{j+1}\) for \(r\le j<t\).  An interior
position has both neighbors in \(I\).  For a binary symbol \(s\), write
\(\bar s=1-s\).
\begin{enumerate}[label=\textup{(\alph*)}]
\item If \(u\) alternates on \(I\), then at every interior position \(i\),
\[
 d(i)=0\quad\Longleftrightarrow\quad v_i=u_i.
\]
If \(v\) also alternates on \(I\), then the two alternating phases either agree
throughout the interior, in which case \(d\equiv0\), or disagree throughout,
in which case every interior position is covered by an exchanged adjacent
pair.
\item Suppose \([i-1,i+3]\subseteq[n]\), the input \(u\) alternates on these
five positions, and the target requires
\[
 v_i=v_{i+1}=v_{i+2}=s.
\]
If \(u_i=s\), no swap set can produce this three-symbol run.  If
\(u_i\ne s\), the unique possible offsets are
\[
 (d(i),d(i+1),d(i+2))=(-1,0,+1).
\]
They use the two flanking input symbols at positions \(i-1\) and \(i+3\)
through the swaps \((i-1,i)\) and \((i+2,i+3)\).  Hence the run is impossible
if either flank is unavailable.  In particular, a target run of at least four
equal symbols cannot be produced when the corresponding input positions and
their one-symbol flanks form an alternating interval.
\end{enumerate}
\end{lemma}
\begin{proof}
For part~(a), alternation gives \(u_{i-1}=u_{i+1}=\bar u_i\) at an interior
position.  Thus offset zero produces \(u_i\), whereas either nonzero offset
produces \(\bar u_i\).  If both \(u\) and \(v\) alternate, they either have the
same phase at every interior position or opposite phases at every interior
position.  Lemma~\ref{lem:domino} then gives the two stated offset patterns.

For part~(b), the symbol \(s\) occupies one parity class in the five-symbol
alternating input.  If \(u_i=s\), producing \(s\) at positions \(i\) and
\(i+2\) forces \(d(i)=d(i+2)=0\), while producing \(s\) at position \(i+1\)
requires \(d(i+1)=\pm1\).  Either nonzero choice would have to be paired with
a nonzero offset at \(i\) or \(i+2\), a contradiction.

If \(u_i\ne s\), the middle output forces \(d(i+1)=0\), while the two outer
outputs require nonzero offsets.  The paired form in
Lemma~\ref{lem:domino} leaves only \(d(i)=-1\) and \(d(i+2)=+1\), using both
flanks.  Finally, a run of four equal output symbols contains two overlapping
three-symbol runs.  Their starting input symbols have opposite values, so one
of them falls into the impossible case \(u_i=s\).
\end{proof}

These lemmas explain the first impossible coordinate in all six continuation
conditions listed above.  A constant target run either cannot be produced from
the alternating transmitted segment, or it requires a flanking input position
that an already forced neighboring swap has consumed.

The family suggests the broader question stated in Section~\ref{sec:open}:
whether every finite directed depth for a two-block binary block set is bounded
linearly in \(|a|+|b|\).  The structural evidence is the
one-sided offset pattern above.  In correction, the two words may use
different swap sets; in detection, only the transmitted word moves relative
to the fixed target, and its nonzero offsets remain confined to disjoint
adjacent pairs.

\section{Extension to \texorpdfstring{\(\LPC(r)\)}{LPC(r)}}
\label{sec:lpcr}

Fix \(r\ge1\).  The radius-one argument used only two facts: a bounded window
determines the committed output prefix, and a common transmitted prefix can be
cancelled.  Both remain true for every fixed displacement radius.  Since
\(B_r\) was defined in Section~\ref{sec:prelim}, set
\[
 \Tball_r(w)=\prefx_{\max\{|w|-r,0\}}(B_r(w)).
\]
A code \(\cC\subseteq\Sigma_q^n\) is \emph{\(r\)-correcting} if
\(B_r(c)\cap B_r(c')=\emptyset\) for distinct \(c,c'\in\cC\), and it is
\emph{\(r\)-detecting} if \(c'\notin B_r(c)\) for distinct \(c,c'\in\cC\).

For \(n\ge1\), equal-length words in \(\Sigma_q^n\) are \emph{\(r\)-live}
when their \(r\)-truncated balls
intersect.  For a finite set \(P\) of nonempty blocks and distinct first blocks
\(a,b\in P\), put
\[
 \Live_n^{(r)}(a,b)=
 \bigl\{\{u,v\}:u\in\Pref_n(aP^*),\ v\in\Pref_n(bP^*),
 \ \Tball_r(u)\cap\Tball_r(v)\ne\emptyset\bigr\}.
\]
Put \(m=\max\{n-r,0\}\).  An ordered pair
\((u,v)\in\Sigma_q^n\times\Sigma_q^n\) is
\emph{directed-\(r\)-live} when
\(\prefx_m(v)\in\Tball_r(u)\), and set
\[
\Live_n^{\to,(r)}(a,b)=
 \{(u,v):u\in\Pref_n(aP^*),\ v\in\Pref_n(bP^*),
          \prefx_m(v)\in\Tball_r(u)\}.
\]
For correction, define the radius-\(r\) extinction depth
\[
 K_r^\star(a,b)=
 \min\{n\ge1:\Live_n^{(r)}(a,b)=\emptyset\},
\]
with value \(\infty\) if the live set never becomes empty.

\begin{lemma}[Window lemma \(W_r\)]
\label{lem:Wr}
If \(c'\in B_r(c)\) and \(r\le\ell\le\abs{c}\), then
\(\prefx_{\ell-r}(c')\in\Tball_r(\prefx_\ell(c))\).
\end{lemma}
\begin{proof}
Write \(c'=\pi c\), with \(\wt(\pi)\le r\).  For
\(j\le\ell-r\), \(\pi(j)\le j+r\le\ell\), so
\(\pi([\ell-r])\subseteq[\ell]\).  List the \(r\) elements of
\([\ell]\setminus\pi([\ell-r])\) as
\(s_1<\cdots<s_r\).  Since their original output positions exceed
\(\ell-r\),
\[
 s_h\ge\ell-2r+h,\qquad s_h\le\ell-r+h;
\]
when \(\ell\le2r\), the first bound is replaced by \(s_h\ge h\).
Define \(\widetilde\pi(j)=\pi(j)\) for \(j\le\ell-r\) and
\(\widetilde\pi(\ell-r+h)=s_h\).  These bounds give
\(\lvert\widetilde\pi(j)-j\rvert\le r\), so
\(\widetilde\pi\in S_\ell\) is banded and
\(\widetilde\pi\prefx_\ell(c)\) agrees with \(c'\) in its first \(\ell-r\)
positions.
\end{proof}

Both correction and directed \(r\)-persistence follow immediately from
Lemma~\ref{lem:Wr}.  The key additional step is cancellation:

\begin{lemma}[Cancellation for every radius]
\label{lem:cancelr}
For a symbol \(\alpha\), if
\(B_r(\alpha s)\cap B_r(\alpha t)\ne\emptyset\), then
\(B_r(s)\cap B_r(t)\ne\emptyset\); common prefixes of any length cancel.
Likewise, if \(c'\in B_r(c)\) and \(c,c'\) have a common prefix, then the
corresponding tails satisfy the same directed relation.  Finally, common
extension preserves both ball intersection and the directed relation: extend
the realizing banded permutations by the identity on the common suffix.
\end{lemma}
\begin{proof}[Proof (rerouting)]
Let \(z=\pi(\alpha s)=\tau(\alpha t)\), and let
\(k=\pi^{-1}(1)\le k''=\tau^{-1}(1)\) (otherwise exchange the two words).
Both lie in \([1,r+1]\), and \(z_k=z_{k''}=\alpha\).  Put
\(m_0=\tau(k)\).  If \(m_0\ne1\), interchange the two assignments
\(\tau(k)=m_0\) and \(\tau(k'')=1\), obtaining \(\tau'\) with
\(\tau'(k)=1\) and \(\tau'(k'')=m_0\).  The output word is unchanged because
both input symbols equal \(\alpha\).  Moreover
\(m_0-k''\le m_0-k\le r\) and
\(k''-m_0\le k''-2\le r-1\), so \(\tau'\) remains banded.

Now \(\pi(k)=\tau'(k)=1\).  Delete output coordinate \(k\) and input
coordinate \(1\), then re-index.  Displacements are unchanged to the right of
\(k\).  For an undeleted assignment \(\pi(j)=i\) with \(j<k\), we have
\(j\le r\) and \(i\ge2\); after the input re-indexing its displacement is
\(i-1-j\ge1-r\), and its upper bound can only improve.  The same argument
applies to \(\tau'\).  Thus every new displacement has absolute value at most
\(r\).  This gives a common output in \(B_r(s)\cap B_r(t)\), and iteration
cancels an arbitrary common prefix.

For the directed statement, write \(c'=\pi c\) and suppose that the two words
begin with the same symbol.  Let \(k=\pi^{-1}(1)\) and \(m_0=\pi(1)\).  Both
belong to \([1,r+1]\).  Interchanging the assignments \(\pi(1)=m_0\) and
\(\pi(k)=1\) leaves the output unchanged, because the two input symbols are
the common first symbol, and the new displacement satisfies
\(|m_0-k|\le r\).  Thus output position \(1\) may be made to use input
position \(1\); delete these matched coordinates and re-index as above.
\end{proof}

Specialized to \(r=1\), this gives alternative short proofs of Lemma~\ref{lem:cancel-import} and of the directed cancellation step in Theorem~\ref{thm:detdepthK}.

\begin{theorem}[Depth-\(K\) criteria for \(\LPC(r)\)]
\label{thm:depthKr}
Fix \(r\ge1\), let \(P\subseteq\Sigma_q^*\) be a finite set of nonempty
blocks, and let \(K\ge r\).
\begin{enumerate}[label=\textup{(\roman*)}]
\item If
\[
 \Live_K^{(r)}(a,b)=\emptyset
 \qquad\text{for every distinct }a,b\in P,
\]
then \(C_n(P)\) is \(r\)-correcting for every \(n\).
\item If
\[
 \Live_K^{\to,(r)}(a,b)=\emptyset
 \qquad\text{for every ordered pair }a\ne b\text{ in }P,
\]
then \(C_n(P)\) is \(r\)-detecting for every \(n\).
\end{enumerate}
\end{theorem}

\begin{proof}
For correction, suppose two distinct codewords have intersecting
radius-\(r\) balls.  Cancel every common initial block using
Lemma~\ref{lem:cancelr}; the remaining pair has distinct first blocks.  If
its length is below \(K\), append the same block sufficiently many times.
Common extension preserves the collision.  Lemma~\ref{lem:Wr} gives
\(r\)-persistence and places the length-\(K\) prefixes in
\(\Live_K^{(r)}(a,b)\), a contradiction.

For detection, start from \(c'\in B_r(c)\), use the directed cancellation
statement, and append a common block if necessary.  The directed form of
Lemma~\ref{lem:Wr} then places the length-\(K\) prefixes in
\(\Live_K^{\to,(r)}(a,b)\), again a contradiction.
\end{proof}

The radius-one configuration graph stored one unresolved end symbol on each
side.  For general radius, that symbol is replaced by a bounded set of input
positions that have not yet supplied an output symbol.  This retains the
finite-state character of the verification problem.

\begin{theorem}[Finite-state verification at fixed radius]
\label{thm:finite-state-r}
Recall that \(N_P=L_P+\lmax\), and put
\[
 V_r(P)=(q+2)^{4r}N_P^2.
\]
For every pair of distinct first blocks \(a,b\in P\), there is a directed
graph \(\mathcal G_r(a,b)\) with at most \(V_r(P)\) vertices such that
\[
 \Live_n^{(r)}(a,b)\ne\emptyset
 \quad\Longleftrightarrow\quad
 \mathcal G_r(a,b)\text{ has a path of length }n
 \text{ from its initial vertex}.
\]
Consequently,
\[
 K_r^\star(a,b)\le V_r(P)+1
 \qquad\text{or}\qquad
 K_r^\star(a,b)=\infty,
\]
and the second alternative holds exactly when the reachable subgraph contains
a directed cycle.  If \(P\) is prefix-free, reachability in the same graph
also decides whether two distinct complete concatenations beginning with
\(a\) and \(b\) have intersecting radius-\(r\) balls.  For arbitrary \(P\),
one additional bit gives the same decision with at most \(2V_r(P)\) vertices.
Thus, for fixed \(q\) and \(r\), finite extinction and exact all-length
correction are decidable in time polynomial in the explicit size of \(P\).
\end{theorem}

\begin{proof}
We describe the information retained after \(j\) transmitted symbols have
been read on one side and the output positions
\(1,\ldots,\max\{j-r,0\}\) have been settled.  Besides the block position
from Definition~\ref{def:one-sided-output}, together with its initial
pre-block state, record the last \(2r\) possible
input positions relative to \(j\).  Each position is marked as already used,
or carries its symbol if it is still pending; positions before the beginning
of the word receive a separate absent mark.  No earlier input can remain
pending, because its last possible output position has already been settled.
Thus one side has at most \((q+2)^{2r}N_P\) states, and the synchronized
two-sided graph has at most \(V_r(P)\) vertices.

A transition appends one parser-legal input symbol on each side.  Once
\(j+1>r\), the next output position to settle is
\(o=j+1-r\).  On each side choose one pending input position
\(i\in[o-r,o+r]\), mark it used, and require the two chosen symbols to be
equal.  If the leftmost pending input is not chosen before it leaves this
interval, the transition is forbidden.  These are exactly the local
conditions for extending two partial banded matchings while producing a
common committed output prefix.

Every pair of radius-\(r\) channel actions witnessing liveness can be revealed
in this order: when output position \(o\) becomes due, its assigned input lies
in \([o-r,o+r]\) and has already been read.  Conversely, a path assigns one
distinct input to every settled output position, always within displacement
\(r\).  At the end of a length-\(n\) path, let
\(p_1<\cdots<p_s\) be the pending input positions, where
\(s=\min\{r,n\}\).  If \(n\ge r\), then
\[
 n-2r+h\le p_h\le n-r+h
 \qquad(1\le h\le r).
\]
Hence matching \(p_h\) to the trailing output position \(n-r+h\) completes
the partial assignment within displacement \(r\); for \(n<r\), use the
identity matching.  Completing both sides independently proves that every
path gives an element of \(\Live_n^{(r)}(a,b)\), establishing the displayed
equivalence.

If a path has length greater than the number of vertices, a vertex repeats.
Repeating the intervening directed cycle gives live pairs at arbitrarily large
depths, and downward persistence gives liveness at every depth.  Otherwise
the reachable graph is acyclic, its path lengths are bounded, and the first
empty level is at most \(V_r(P)+1\).

Finally, a vertex reached when both parsers are at block boundaries represents
a full collision exactly when the two pending sets admit banded matchings to
the trailing output positions that produce the same trailing word.  This is a
finite predicate on the recorded frontier.  If \(P\) is prefix-free, distinct
first blocks already guarantee distinct transmitted words.  In general, one
additional bit records whether the two transmitted prefixes have differed in
some coordinate and excludes two parsings of the same word.  A depth-first
search therefore decides both cycle existence and collision reachability.
\end{proof}

\noindent Theorem~\ref{thm:depthKr} shows that the one-condition architecture
is not a radius-one phenomenon: for every fixed displacement radius, one
finite horizon certifies all codeword lengths.  What changes for \(r\ge2\) is
the local normal form.  Truncation by \(r\) symbols need not be injective on
\(B_r\), so the \(\abs{\Tball}=\abs{\B}\) identity and closest-witness
reformulation of~\cite{PriorWork2026} no longer apply; this is why the theorem
is stated directly using \(\Tball_r\).  Theorem~\ref{thm:finite-state-r}
supplies a general finite frontier, but its sharp size and the periodic normal
forms behind long-lived pairs remain open for \(r\ge2\).  Analogues of the
\(\Dmin\ge2\) shortcut and the single-length identity
of~\cite[Thm.~5]{PriorWork2026} are also open.
A related design principle is to use constant buffers to isolate the cut
interface, as in the structured constructions of~\cite[Cor.~6]{PriorWork2026}.


\section{Open problems}
\label{sec:open}

The results above leave several focused questions.
\begin{enumerate}[leftmargin=*]
\item \textbf{Two-block extinction depth.}
Prove Conjecture~\ref{conj:B2}.  The main missing step is to turn the
Fine--Wilf overlap mechanism into a bound on the longest live path, rather
than merely on the number of reachable states.

\item \textbf{Linear directed depth.}
Decide whether \(\dko(a,b)<\infty\) always implies
\(\dko(a,b)=O(|a|+|b|)\) for binary two-block sets, and whether an
\(O(L_P)\) bound holds for general finite block sets.

\item \textbf{Large-alphabet correction loss.}
Determine whether \((1-C_0^{(q)})\ln q\) converges as \(q\to\infty\), and
close the interval between \(\ln\varphi\) and \(\ln(1.82560995)\) from
Corollary~\ref{cor:correction-sandwich}.  On the upper-bound side this requires
using repeated-symbol components more efficiently than the all-distinct lift;
on the constructive side it requires reducing the run-wise spectral loss.

\item \textbf{Fixed-alphabet detection.}
Theorem~\ref{thm:detqary} determines the large-alphabet loss to first order,
but the exact detecting capacity is open for every fixed \(q\).  In
particular, for \(q=3\) the present interval is
\(0.736917768\le R_{\mathrm{det},3}\le\frac12\log_3 6\).

\item \textbf{Infinite-gap classification.}
Prove the two-block rigidity suggested by the census: must the two period
labels of every two-block gap witness be cyclic rotations of one another?
Classify the failures of this rotation relation that first appear for three
or more blocks, and give a word-theoretic characterization that does not pass
through the configuration graph.

\item \textbf{Higher displacement radius.}
For \(r\ge2\), determine the sharp finite-state frontier size and the
periodic structures governing long-lived pairs.

\end{enumerate}

\section{Conclusion}
\label{sec:conclusion}

Extinction depth gives a mathematical layer between the two-stage criterion
of~\cite{PriorWork2026} and exact verification.  The central theorem has one
condition: if every pair beginning with distinct first blocks is extinct at one
common depth, then every finite-length concatenation code is correcting.  The
common-extension argument explains why no additional short-collision clause is
needed.

The parameter is nontrivial.  Proposition~\ref{prop:linear-depth} gives exact
depth \(L+2\) in an elementary family, so no bound can ignore the block
lengths.  At the same time, Theorem~\ref{thm:single-length-pumping} gives the
linear cutoff \((4q^2+1)\ell\) for every single-length block set.  More
generally, finite depth admits the ranking certificates of
Proposition~\ref{prop:rank}, while failure of finite extinction for a
prefix-free correcting block set is exactly an infinite nonclosing collision.
On the corridor family, exact finite-state verification proves the quadratic
formula \(2k^2+4k-1\) for every \(3\le k\le60\); extending this formula to all
\(k\ge3\) remains a conjecture.

As a construction tool, Theorem~\ref{thm:rates} yields a ternary block set of
rate above \(0.6777475\), while Theorem~\ref{thm:q5plus} gives a separate
five-ary template construction of rate above \(0.6694926\).
On the upper-bound side, Theorem~\ref{thm:length-six-cover} improves the
finite-alphabet covering bounds, while Theorem~\ref{thm:length-eighteen-cover}
strengthens them by refining every two-symbol component simultaneously.
Theorem~\ref{thm:covering-phi} places the normalized large-alphabet correction
loss above \(\ln\varphi\).
Theorem~\ref{thm:runwise-main} places it below \(\ln(1.82560995)\), reducing
the constructive loss factor below \(1.82560995\).  For detection, stable type
lifting gives the first-order optimal factor \(\sqrt2\), while directed
extinction proves the exact certification depth \(4k\) on the corridor family.
Finally, Theorem~\ref{thm:depthKr} extends both depth criteria to every fixed
displacement radius \(r\), and Theorem~\ref{thm:finite-state-r} gives a finite
correction-verification graph at every such radius.

All finite verifications used in the proofs were performed by exhaustive
enumeration or exact integer arithmetic, and every displayed numerical root
was bounded by exact rational sign checks.  The available verifiers, finite
certificates, canonical outputs, and the precise scope of each archived check
are recorded in the accompanying Zenodo reproducibility
package~\cite{ReproPackage2026}.
\begin{appendices}
\renewcommand*{\theHtable}{appendix.\Alph{section}.\arabic{table}}

\section{Explicit block sets}
\label{app:blocks}

The full \(265\)-block set \(P^{(3)}_{14}\) is listed in
Table~\ref{tab:p3-265-blocks}, grouped by length.  The base set \(P^{(3)}\)
consists of the blocks of lengths \(3,4,6,8,10\); the remaining rows are the
blocks added to form \(P^{(3)}_{14}\).
\begingroup
\scriptsize
\setlength{\tabcolsep}{3pt}
\renewcommand{\arraystretch}{0.96}

\begin{longtable}{@{}c>{\ttfamily}l>{\ttfamily}l>{\ttfamily}l@{}}
\caption{The \(265\) blocks of \(P^{(3)}_{14}\), grouped by length.}
\label{tab:p3-265-blocks}\\
\toprule
Length & \multicolumn{3}{c}{Blocks} \\
\midrule
\endfirsthead

\toprule
Length & \multicolumn{3}{c}{Blocks} \\
\midrule
\endhead

\bottomrule
\endfoot

3 & 000 & 111 & 222 \\
\midrule
4 & 0111 & 0122 & 0222 \\
4 & 1000 & 1222 & 2000 \\
4 & 2100 & 2111 &  \\
\midrule
6 & 001111 & 001222 & 002111 \\
6 & 002222 & 011022 & 012011 \\
6 & 022011 & 110000 & 110222 \\
6 & 112000 & 112222 & 122100 \\
6 & 201111 & 210122 & 220000 \\
6 & 220111 & 221000 & 221111 \\
\midrule
8 & 00111200 & 00112000 & 00211200 \\
8 & 00212000 & 00222100 & 01201200 \\
8 & 01210000 & 11000211 & 11200211 \\
8 & 11222011 & 21021111 & 22000122 \\
8 & 22010222 & 22110222 & 22111022 \\
\midrule
10 & 0011120211 & 0011122200 & 0021120211 \\
10 & 0021122200 & 0022210211 & 0022210222 \\
10 & 0120000211 & 0120000222 & 0120102211 \\
10 & 0120102222 & 1100021200 & 1100211111 \\
10 & 1120211111 & 1122201200 & 2101210000 \\
10 & 2101210022 & 2102221111 & 2102222000 \\
10 & 2200012000 & 2200012100 & 2211102000 \\
10 & 2211102100 &  &  \\
\midrule
11 & 10211112221 & 11022110222 & 22100212000 \\
\midrule
12 & 001110000111 & 001110002222 & 001112012200 \\
12 & 001112012222 & 001112022211 & 001112022222 \\
12 & 001112222111 & 001200112000 & 001220011200 \\
12 & 001221012200 & 001221012222 & 002112022211 \\
12 & 002112022222 & 002112222111 & 002221012200 \\
12 & 002221012222 & 012000022100 & 012000122212 \\
12 & 012001122000 & 021112000111 & 021112220000 \\
12 & 021112222000 & 021112222111 & 102111022212 \\
12 & 102111100022 & 102111120000 & 102111120011 \\
12 & 102111122000 & 102111122011 & 110002102221 \\
12 & 110002221111 & 110002222000 & 110021110222 \\
12 & 110201112000 & 110201112011 & 110221000222 \\
12 & 110221100211 & 110221101200 & 110221110000 \\
12 & 110221111222 & 112001222000 & 112002120000 \\
12 & 112021110222 & 112220102221 & 122110201111 \\
12 & 210121000211 & 210222100000 & 210222100011 \\
12 & 210222110222 & 210222200122 & 220001121111 \\
12 & 220001201111 & 221110000211 & 221110000222 \\
12 & 221110002211 & 221110002222 & 221110211111 \\
\midrule
13 & 0011200112000 & 0012001120111 & 0012211112221 \\
13 & 0012211122200 & 0022211112221 & 0022211122200 \\
13 & 0120001222222 & 0120002221111 & 0120002222000 \\
13 & 0121000112000 & 1021110222000 & 1021110222222 \\
13 & 1021111222200 & 1021111222222 & 1021112220000 \\
13 & 1021112220011 & 1021112222000 & 1021112222111 \\
13 & 1100021022111 & 1122201022111 & 1221102011222 \\
13 & 2200011110000 & 2200012011222 & 2201022010222 \\
13 & 2201120211111 & 2210021200122 & 2211022010222 \\
\midrule
14 & 00111000012000 & 00111000012222 & 00111000022201 \\
14 & 00111000022222 & 00111000212000 & 00111000212222 \\
14 & 00111201112000 & 00111201112011 & 00111201220211 \\
14 & 00111202102222 & 00111202221200 & 00111211100022 \\
14 & 00111222011111 & 00111222012222 & 00111222201111 \\
14 & 00111222201122 & 00112001222000 & 00122101112000 \\
14 & 00122101112011 & 00122101220211 & 00122102102222 \\
14 & 00122111022222 & 00122111100022 & 00122111120000 \\
14 & 00122111120011 & 00122111122000 & 00122111122011 \\
14 & 00211202221200 & 00211222011111 & 00211222012222 \\
14 & 00211222201111 & 00211222201122 & 00212001222000 \\
14 & 00212002222000 & 00222101220211 & 00222102102222 \\
14 & 00222111022222 & 00222111100022 & 00222111120000 \\
14 & 00222111120011 & 00222111122000 & 00222111122011 \\
14 & 01102001122000 & 01200012220111 & 01200211100000 \\
14 & 01210001222000 & 01210011122200 & 02111200012000 \\
14 & 02111200012222 & 02111202120000 & 02111222010222 \\
14 & 10211102110222 & 10211102220111 & 10211102220211 \\
14 & 10211102220222 & 10211110000211 & 10211112201022 \\
14 & 10211121002222 & 11000210222200 & 11000210222222 \\
14 & 11000222200122 & 11001110002222 & 11001112222111 \\
14 & 11002110000111 & 11022111100000 & 11022111100011 \\
14 & 11200111100000 & 11200111122222 & 11200111222222 \\
14 & 11200122201111 & 11200210210000 & 11200212000122 \\
14 & 11202110000111 & 11222001122000 & 11222010222200 \\
14 & 11222010222222 & 11222022110222 & 12211020111200 \\
14 & 12211020111222 & 12211020112000 & 21012100211111 \\
14 & 21021112000111 & 21021112002222 & 21021112220000 \\
14 & 21021112222000 & 21021112222111 & 21022011100000 \\
14 & 21022200112221 & 21022210001200 & 21022220012100 \\
14 & 21120111122200 & 21120111200000 & 22000111000220 \\
14 & 22000111122200 & 22000111200000 & 22000120111200 \\
14 & 22000120111222 & 22000120112000 & 22000121110222 \\
14 & 22010220000222 & 22010220111222 & 22010221001111 \\
14 & 22011202120000 & 22110220000222 & 22110220110000 \\
14 & 22110220110200 & 22110220111222 & 22110221001111 \\
14 & 22111000021200 & 22111000212000 & 22111000221200 \\
14 & 22111021110222 &  &  \\

\end{longtable}
\endgroup

\end{appendices}

\end{document}